\documentclass[11pt,letterpaper]{article}

\usepackage{amsmath,amsthm,nicefrac}
\usepackage{amsfonts, amstext}
\usepackage{bm}

\usepackage{subcaption}
\usepackage{hhline}
\usepackage{soul}
\usepackage[table,dvipsnames]{xcolor}

\usepackage[letterpaper,top=1in,bottom=2cm,left=1in,right=1in,marginparwidth=1.75cm,footskip=1cm]{geometry}

\usepackage[font={small,it}]{caption}

\usepackage{setspace}

\usepackage[shortlabels]{enumitem}
\usepackage[breaklinks]{hyperref}
\hypersetup{colorlinks=true,
            citebordercolor={.6 .6 .6},linkbordercolor={.6 .6 .6},
citecolor=blue,urlcolor=blue,linkcolor=blue,pagecolor=black,breaklinks}

\usepackage[nameinlink]{cleveref}
\Crefname{algocf}{Algorithm}{Algorithms}
\crefname{algocfline}{line}{lines}
\Crefname{invariant}{Invariant}{Invariants}
\Crefname{claim}{Claim}{Claims}
\Crefname{corollary}{Corollary}{Corollaries}
\Crefname{subclaim}{Subclaim}{Subclaims}

\usepackage{epsfig}
\usepackage{amsthm,amssymb}

\usepackage{mathrsfs}
\usepackage{xspace}
\usepackage{soul}
\usepackage{latexsym}
\usepackage{bbm}
\usepackage{dsfont}

\usepackage{accents}

\usepackage{framed}

\definecolor{DarkGray}{rgb}{0.66, 0.66, 0.66}
\definecolor{DarkPowderBlue}{rgb}{0.0, 0.2, 0.6}
\definecolor{fluorescentyellow}{rgb}{0.8, 1.0, 0.0}
\definecolor{cerulean}{rgb}{0.0, 0.48, 0.65}
\definecolor{bleudefrance}{rgb}{0.19, 0.55, 0.91}

\usepackage[ruled,vlined,linesnumbered,algonl,procnumbered]{algorithm2e}
\SetEndCharOfAlgoLine{}
\SetKwComment{Comment}{\footnotesize$\triangleright$\ }{}

\SetCommentSty{mycommfont}

\usepackage{thmtools,thm-restate}

  \usepackage[skip=5pt plus 3pt minus 1pt,parfill=0pt]{parskip}
\makeatletter
  \def\thm@space@setup{\thm@preskip=\parskip \thm@postskip=0pt
  }
  \makeatother

\allowdisplaybreaks

\newtheorem{theorem}{Theorem}[section]
\newtheorem{lemma}[theorem]{Lemma}
\newtheorem{claim}[theorem]{Claim}

\newtheorem{observation}[theorem]{Observation}

\theoremstyle{definition}
\newtheorem{defn}[theorem]{Definition}

\theoremstyle{remark}
\newtheorem{remark}[theorem]{Remark}

\usepackage{tikz}
\usepackage{pgfplots}
\pgfplotsset{compat=newest}
\usetikzlibrary{arrows.meta,calc,decorations.pathreplacing,positioning}

\definecolor{freshgreen}{HTML}{009e74}\definecolor{freshknown}{HTML}{56b3e9}\definecolor{ongoingred}{HTML}{d55e00}
\definecolor{processblue}{RGB}{38,86,210}
\definecolor{processfill}{HTML}{e69f00}
\definecolor{gridgray}{RGB}{215,215,215}

\definecolor{col1}{HTML}{56b3e9}
\definecolor{col2}{HTML}{e69f00}
\definecolor{col3}{HTML}{009e74}
\definecolor{col4}{HTML}{cc79a7}
\definecolor{col5}{HTML}{d55e00}
\definecolor{col6}{HTML}{0071b2}

\providecommand{\hdots}{\ldots}

\tikzset{
vertex/.style={circle, draw, fill=black, black, inner sep=0pt, minimum width=9pt},
asg/.style={line width=1pt},
algCell/.style={draw=gridgray,line width=.35pt},
algProcCell/.style={fill=processfill!50!white,fill opacity=.62,line width=.85pt},
algProcDot/.style={circle,draw=processblue,line width=.9pt,inner sep=0pt,minimum size=7.2pt},
algStayArrow/.style={dashed,draw=black!45,line width=.55pt},
}

\newcommand{\algFreshDot}[4]{\fill[freshgreen] ($ (#1,#2)+(#3,#4) $) circle[radius=3pt];}
\newcommand{\algFreshKnownDot}[4]{\fill[freshknown] ($ (#1,#2)+(#3,#4) $) circle[radius=3pt];}
\newcommand{\algOngoingDot}[4]{\fill[ongoingred] ($ (#1,#2)+(#3,#4) $) circle[radius=3pt];}
\newcommand{\algOptionalDot}[5]{\draw[ongoingred,dashed,line width=.8pt] ($ (#1,#2)+(#3,#4) $) circle[radius=2.55pt];
  \node[anchor=west,inner sep=1pt,text=black!70] at ($ (#1,#2)+(#3+.10,#4-.15) $) {\tiny $#5$};
}

\usepackage{comment}

\usepackage[
 backend=biber,
 style=alphabetic,
 citetracker,
 hyperref=auto,
 maxcitenames=5,
 sortcites,
 sorting=nyt,
 maxbibnames=12,
 date=year,
 isbn=false,
 url=false,
 doi=false,
 eprint=false,
]{biblatex}

\newbibmacro{string+doiurlisbn}[1]{
  \iffieldundef{doi}{
    \iffieldundef{url}{
      #1
    }{
      \href{\thefield{url}}{#1}
    }
  }{
    \href{http://dx.doi.org/\thefield{doi}}{#1}
  }
}

\DeclareFieldFormat
[article,inbook,incollection,inproceedings,patent,thesis,unpublished]
  {title}{\usebibmacro{string+doiurlisbn}{#1}}

\usepackage{titlesec}
\usepackage{lipsum}
\usepackage{float}
\usepackage{csquotes}

\usepackage[normalem]{ulem}

\newcommand{\eps}{\varepsilon}

\newcommand{\ZZ}{\mathbb{Z}}

\newcommand{\opt}{\ensuremath{\mathrm{OPT}}\xspace}
\newcommand{\OPT}{\opt}

\renewcommand{\emptyset}{\varnothing}

\newcommand{\floor}[1]{\lfloor#1\rfloor}

\newcommand{\ts}{{t^*}}

\newcommand{\afull}{\smash{A^{\mathrm{fresh}}}}
\newcommand{\apart}{\smash{A^{\mathrm{on}}}}
\newcommand{\aon}{\smash{A^{\mathrm{on}}}}

\newcommand{\fink}{k^{\mathrm{f}}}
\newcommand{\fint}{t^{\mathrm{f}}}
\newcommand{\redp}{p^{\mathrm{e}}}
\newcommand{\redk}{k^{\mathrm{e}}} 
\newcommand{\redex}{\mathrm{ex}^{\mathrm{e}}}
\newcommand{\ex}{\mathrm{ex}}
\newcommand{\lppart}{\mathrm{LP}}
\newcommand{\dppart}{\mathrm{DP}}

\newcommand{\cafull}{\smash{\mathrm{Q}^{\mathrm{fresh}}}}
\newcommand{\capart}{\smash{\mathrm{Q}^{\mathrm{on}}}}
\newcommand{\caon}{\smash{\mathrm{Q}^{\mathrm{on}}}}

\newcommand{\EX}{\mathbb{E}}
\newcommand{\VAR}{\mathbb{V}}

\newcommand{\prob}{\mathbb{P}}

\def\marrow{\marginpar[\hfill$\longrightarrow$]{$\longleftarrow$}}

\newcommand{\nf}{\nicefrac}

\usepackage{xcolor}
\definecolor{job1}{HTML}{6D9DC5}
\definecolor{job2}{HTML}{E0CC74}
\definecolor{job3}{HTML}{649884}
\definecolor{job4}{HTML}{C9809C}
\definecolor{job5}{HTML}{F39B6D}
\definecolor{job6}{HTML}{547781}
\definecolor{job7}{HTML}{1c8591}

\newcommand\machine{}
\def\machine[#1](#2,#3)(#4,#5){
  \draw[black,#1,line width=2pt] (#4,#3) -- (#2,#3) -- (#2,#5) -- (#4,#5);
}

\newcommand\interval{}
\def\interval[#1](#2,#3){
  \draw[line width=1pt,dotted,black,#1] (#2,0) -- (#2,#3);
}

\newcommand\rectjob{}
\def\rectjob[#1](#2,#3)(#4,#5){
  \draw[black,fill=#1,line width=0.7pt] (#2,#3) rectangle ++(#4,#5);
}

\newcommand\rectjobT{}
\def\rectjobT[#1](#2,#3)(#4,#5)(#6){
  \draw[black,fill=#1,line width=0.7pt] (#2,#3) rectangle ++(#4,#5);
  \node[#1,black] at (#2+#4/2.0,#3+#5/2.0) {{#6}};
}

\newcommand\matchingjob{}
\def\matchingjob(#1,#2)(#3,#4){
  \draw[black,fill=none,line width=1.2pt] (#1,#2) rectangle ++(#3,#4);
}

\newcommand\partialjob{}
\def\partialjob(#1,#2)(#3,#4){
  \draw[black,dashed,fill=none,line width=1.2pt] (#1,#2) rectangle ++(#3,#4);
}

\newcommand\rectjobnob{}
\def\rectjobnob[#1](#2,#3)(#4,#5){
  \draw[white,fill=#1,line width=0.0pt] (#2,#3) rectangle ++(#4,#5);
}

\def\marrow{\marginpar[\hfill$\longrightarrow$]{$\longleftarrow$}}
\def\anupam#1{\textsc{\color{magenta} (Anupam says: }\marrow\textsf{\color{magenta} #1)}}

\title{Better Late Than Never: \\ Online Flow Time Scheduling with Online Estimates}

\author{
   Anupam Gupta\thanks{
       New York University, \texttt{anupam.g@nyu.edu}. 
   } \and Haim Kaplan\thanks{
       Tel Aviv University,
       \texttt{haimk@tau.ac.il}
   }  \and Alexander Lindermayr\thanks{
       Institut für Mathematik, Technische Universität Berlin,
       \texttt{alexander.lindermayr@tu-berlin.de}
   }  \and Jens Schlöter\thanks{
       Department of Mathematics and Computer
Science, University of Southern Denmark,
       \texttt{schloter@imada.sdu.dk}
   } \and
        Sorrachai Yingchareonthawornchai\thanks{
       Institute for Theoretical Studies, ETH Zürich,
       \texttt{sorrachai.yingchareonthawornchai@eth-its.ethz.ch}
   }}

\date{}
\begin{document}

\maketitle

\begin{abstract}

In the classical online flow-time scheduling problem on a single machine, 
jobs arrive over time and must be processed in some order to minimize the total time they spend in the system: for over fifty years, we have known that SRPT is an optimal online algorithm. But this algorithm requires exactness in two different ways: (a) job sizes must be known exactly, and (b) they must be revealed as soon as the job arrives. Recent work relaxed each of these assumptions separately: there are algorithms based on knowing approximate sizes (given when the job arrives), or based on knowing (exact) sizes at some point before the remaining size gets too small. (The latter setting models the phenomenon that the algorithm may ``learn'' job sizes by processing them---but asking for learning algorithms to output exact sizes is unreasonable.) Nonetheless, prior to this work, there was no known approach to relax both of these assumptions simultaneously.

In this work, we consider a model that demands much less: jobs arrive over time. When a job arrives, we are notified. As we process it,
at some point in time between when we complete an $\eps$-fraction and
a $(1-\eps)$-fraction of its unknown processing requirement, we are
informed that the job is ``somewhere in the middle''. Finally, when the
job has received its desired amount of processing, we are informed of
its completion. No other information is shared about the job. We give an $O(1/\eps^2)$-competitive algorithm for this model.

Slightly more generally, we assume that 
an algorithm receives a $\mu$-approximate estimate of each job's processing time at some time before we complete a $(1-\eps)$-fraction of its processing.   
Our algorithm is $O(\mu/\varepsilon)$-competitive, and we show that this is asymptotically optimal. (The previous ``three bit'' model can be solved using this algorithm with $\mu = 1/\eps$.) Our algorithm is a surprisingly natural variant of the multilevel feedback algorithm (MLF), which is widely used in practice, and is parameter-oblivious: it does not need to know $\mu$ or $\varepsilon$ upfront. The core analytical contribution is to generalize and robustify the dual-fitting framework for this problem to handle jobs for which we have not yet received estimates.

 \end{abstract}

\thispagestyle{empty}

\newpage

\setcounter{page}{1}

\section{Introduction}
\label{sec:introduction}

We consider the problem of minimizing the total flow time of jobs
arriving online on a single machine. This is a classical problem for
which the textbook algorithm \emph{Shortest Remaining Processing Time
  (SRPT)} is optimal \cite{Schrage68} (i.e., $1$-competitive). As the name suggests, at
each time the algorithm processes a job which has the least amount of
processing remaining to be done. However, this algorithm relies on two
assumptions:
\begin{enumerate}[label=(A\arabic*)]
\item \label{item:1} \emph{the precise job size $p_j$ for each job $j$ is known}, and 
\item \label{item:2} \emph{the job size is revealed as soon as the job arrives}.
\end{enumerate}
While this is reasonable for some settings, other settings may violate
one or both of these assumptions. The nonclairvoyant model captures
an extreme version of this uncertainty: it assumes that size is
revealed only when the job completes. While SRPT is no longer
applicable, other natural algorithms, like \emph{Round-Robin},
\emph{Shortest Elapsed Time First (SETF)}, or its close cousin,
\emph{Multilevel Feedback (MLF)} are all applicable. However, proving
constant-competitiveness in this model is impossible:
$\Omega(n^{1/3})$-competitive deterministic algorithms and
$\Omega(\log n)$-competitive randomized algorithms being the best
possible \cite{MotwaniPT94}, where $n$ is the total number of
jobs. Given the bleakness of these competitive ratios, we ask: is
there a little more information we can use to get better algorithms?
In the past decades, numerous
works~\cite{BecchettiLMP04,YingchareonthawornchaiT17,AzarLT21,AzarLT22,AzarPT22,GuptaKLSY25,BenomarCLS25,GuptaKPW26,Lindermayr26}
have sought to address precisely this question. However, each relaxes
either \ref{item:1} or \ref{item:2}, but none address both simultaneously.

We propose the following simple \emph{$\eps$-signaling} model:
\begin{quote}
  Jobs arrive over time. When a job arrives, we know nothing about it
  (yet). As we process it, we receive a \emph{signal}. The timing of
  the signal is allowed to be at any time (adversarially) between the times
  when we complete an $\eps$-fraction and a $(1-\eps)$-fraction of its
  unknown processing requirement. Of course, we can measure the
  elapsed time ourselves, and we get to know when we finish processing the job.
  Those are the only pieces of information we get about the job---the information that a job has arrived, the information that we are currently 
  ``somewhere in the middle'' of its processing, and finally, the information
   that the job has just completed. 
\end{quote}

This model requires little clairvoyance from the scheduling
algorithm. Given a weak job profiler (or job estimator) that only
tells us whether we have processed between $\eps$ and $1-\eps$ of the
job's size at some point, can we give good algorithms in this setting,
especially one that does not 
need to know $\eps$ either? To our surprise, we show:

\begin{theorem}[Algorithm for Signaling]
  \label{thm:signal}
  There exists an algorithm for total flow time
  minimization in the $\eps$-signaling model that, for any $\eps \in (0,\nf12]$,  ensures a total flow
  time within a factor of $O(\nf{1}{\eps^2})$ of the optimal value. The
  algorithm does not need any knowledge of $\eps$. 
\end{theorem}

Previously, an $O(\nf{1}{\eps})$-competitive algorithm was shown in the setting where a job's \emph{exact} size is revealed when \emph{exactly} a $(1-\eps)$-fraction of its processing time is done~\cite{GuptaKLSY25}. Our result successfully removes the impractical assumption that a job profiler has to be precise. 

This model and theorem will be a consequence of our actual model of
interest: in the \emph{$(\eps,\mu_1,\mu_2)$-online estimate model},
for each job $j$, the adversary reveals a size estimate
$\hat{p}_j \in [\nf{p_j}{\mu_2}, \mu_1\cdot p_j]$ to the algorithm at some
point in time while at least an $\eps$-fraction of the job $j$ still remains to
be processed. Our main result is as follows:

\begin{theorem}[Algorithm for Online Estimates] \label{thm:main} There is an algorithm
for total flow time minimization in the online estimate model that,
for every $\mu_1,\mu_2 \geq 1$ and $\eps \in (0,1]$, ensures a total flow
  time within a factor of $O(\frac{\mu_1 \mu_2}{\eps})$ of the optimal value.
  It does not require $\mu_1, \mu_2$, or $\eps$ as inputs.
\end{theorem}

(For completeness, in \Cref{app:reduction} we prove that \Cref{thm:main}
implies \Cref{thm:signal}.) The algorithm of \Cref{thm:main}, which we call Balanced MLF, is
``hyperparameter-oblivious'': it does not need to know any of
$\mu_1, \mu_2$, or $\eps$. It only performs $O(\sum_j (1+\log_2 p_j))$
preemptions overall (the proof is in \Cref{app:preemptions}). Moreover, it is conceptually simple, close to the
widely-used MLF algorithm, and easy to describe and implement, as we
will see next.

\paragraph{Balanced MLF (informal).}
There are multiple queues (classes) $Q_0,Q_1,\ldots$ of jobs. A job can be
\emph{fresh} or \emph{ongoing}.  When a job arrives, it is fresh and
placed in $Q_0$. Whenever it has received $2^{k}$ units of processing, and its estimate has not yet arrived, it is promoted to queue
$Q_{k+1}$. When the estimate $\hat p_j$ of job $j$ arrives while it
is in queue $Q_k$, it is promoted to queue $Q_{k'}$ with
$k' = \lfloor \log_2 \hat p_j \rfloor$ if $k' > k$; otherwise it stays
ongoing in its current queue.\footnote{We assume w.l.o.g.\ that the estimate arrives only when the algorithm processes a job. Otherwise, the algorithm can defer using the estimate until it processes a job.} Whenever a job is promoted, it becomes
fresh again.  At each time step, if at least a $\nicefrac14$
fraction of active jobs is fresh, make the first job of the
smallest-index nonempty queue ongoing if it is currently fresh;
process the ongoing job with the smallest-index queue for one unit of
time. (We defer a discussion about intuition of the algorithm to
\Cref{sec:tech-overview} and a formal description to \Cref{sec:alg}.)

Besides our algorithmic results, we view the definition of the
\emph{online estimate model} as a central contribution of this
work.   Indeed, it captures and unifies several models for robustness
and noisy/partial information that are popular in the literature:
\begin{enumerate}[label=(\alph*)]
\item The \emph{robust flow time} or \emph{semi-clairvoyant} models
  considered
  by~\cite{BecchettiLMP04,AzarLT21,AzarLT22,AzarPT22,GuptaKPW26},
  which are special cases of our model where $\eps = 1$, and the
  algorithm receives a noisy estimate as soon as the job arrives. Our algorithm (Balanced MLF) is identical to their algorithm (Balance) in \cite{GuptaKPW26} when we restrict $\eps = 1$. 

\item The \emph{$\eps$-clairvoyance model}
  of~\cite{YingchareonthawornchaiT17,GuptaKLSY25,BenomarCLS25}, which
  is another special case of our model where $\mu_1 = \mu_2 = 1$, and
  the algorithm receives an exact size-estimate before the remaining
  fraction of its processing drops below $\eps$.  Balanced MLF is also asymptotically optimal in this model.\item Moreover, algorithms in the $\eps$-clairvoyance model imply
  nonclairvoyant algorithms in the
  \emph{$(1+\eps)$-speed-augmentation} model (see \cite{GuptaKLSY25} for a reduction)
  which gives a connection for yet another well-studied model~\cite{KalyanasundaramP00,ImMP11,Edmonds00}. \end{enumerate}
The first two models relax one of the assumptions \ref{item:1} and
\ref{item:2} we stated at the beginning of the paper. In this work, we
manage to relax both at once. The last model is nonclairvoyant, but
gives a bi-criteria guarantee, comparing apples to oranges. Given that our
model subsumes these previous ones, our algorithm (Balanced MLF) is
asymptotically optimal for these models simultaneously.

Finally, our matching lower bound
shows that the competitiveness is asymptotically optimal:

\begin{theorem}[Lower Bound] \label{thm:lower-bound} For every
  $\mu_1,\mu_2 \geq 1$ and $\eps \in (0,1]$, every (randomized)
  algorithm has a competitive ratio in
  $\Omega(\frac{\mu_1 \mu_2}{\eps})$ for minimizing the total flow
  time on a single machine in the $(\eps,\mu_1,\mu_2)$-online
  estimate model, even if $\mu_1,\mu_2$, and $\eps$ are given up-front to the
  algorithm.
\end{theorem}

\subsection{Technical Overview}\label{sec:tech-overview}

From a technical standpoint, our algorithm \emph{Balanced MLF}
faces various challenges to achieving a constant competitive ratio in the online estimate model.

Our algorithm draws on several previous algorithms (MLF~\cite{Bach86}, SLF~\cite{GuptaKLSY25}, and
\textsc{Balance}~\cite{GuptaKPW26}) and combines them together to achieve the best of
all of them. In particular:
\begin{itemize} 
\item (Nonclairvoyant algorithm) As the name suggests, the base of our
  algorithm is the nonclairvoyant \emph{Multilevel Feedback} (MLF)
  algorithm, which is the classical backbone of modern computer
  schedulers~\cite{Bach86,CoffmanD73,madnick1974operating,SGG2018}.
In MLF, each job gets $2^{k}$ units of processing while in queue
  $Q_k$ before being promoted to queue $Q_{k+1}$. At any time, the
  algorithm runs the front job of the lowest nonempty queue.
\item (Known and unknown jobs). While MLF is not competitive by
  itself, it gives us a framework to explore \emph{unknown} jobs
  (i.e., those whose estimates have not yet arrived), while also
  naturally grouping jobs into classes induced by queues. We build on
  the MLF framework as follows: Whenever we receive a job's estimate,
  the job becomes \emph{known}; we fix its class based on the
  estimate and never promote the job again. This means that we
  balance the decision between exploring unknown jobs and exploiting
  our knowledge of (approximate) shortest jobs when a known job is in
  the smallest class, in a similar spirit to the algorithm in the $\eps$-clairvoyant
  setting~\cite{GuptaKLSY25}.
\item (Partial and full jobs). We call jobs \emph{fresh} when they are
  promoted to a new class and \emph{ongoing} once we started to work
  on them in their current class. The concept of partial and full jobs
  (i.e.\ jobs we have worked on and jobs we have not touched yet after their arrival) in \cite{BecchettiL04,AzarLT22,GuptaKPW26} becomes ongoing and fresh
  jobs, respectively, in the online estimate model because our classes
  (i.e. MLF queues) are \emph{dynamic}. To avoid having too many
  ongoing jobs, we use the balancing rule of \cite{GuptaKPW26} to
  ensure that only a constant fraction of active jobs is ongoing at
  any time.
\end{itemize}
Of course, the challenge is to analyze its performance and show that the combined algorithm is competitive in the online estimate setting.

\medskip\textbf{The Major Challenges.}  The existing analyses do not
seem to be directly applicable to our setting for two fundamental
reasons: (1) \emph{mutability of classes} and (2) \emph{the absence of
  full jobs}. 
The job mutability comes from the properties of our
algorithm. In the setting of~\cite{GuptaKPW26}, estimates are
available when jobs are released, and the relevant class of a job is
essentially fixed during its lifetime. In contrast, Balanced MLF must
explore unknown jobs before their estimates arrive, and this
exploration can change their classes multiple times.

However, the lack of full jobs is inherent in the model: the jobs that
our algorithm keeps under control are not---and indeed, cannot
be---full jobs. Previous analyses
\cite{BecchettiL04,AzarLT21,AzarLT22,GuptaKPW26} crucially rely on
full jobs, i.e., jobs that the algorithm has not yet processed,
because a full job satisfies $p_j(\ts)=p_j$, so its remaining volume
can be compared directly with its estimate $\hat p_j$ using the
distortion parameter $\mu$. This is one reason why these algorithms
maintain a balance between full and partial jobs, control the number
of full jobs, and use this number as a surrogate for the total number
of jobs. However, in our setting, all jobs have to be
partial---indeed, we need to process a job before we learn something
about it. For Balanced MLF, the natural analogs of full jobs are fresh
jobs. But a fresh job is likely not full: it may already have received
a large amount of exploratory processing before it last became
fresh. Thus, its true remaining processing time may be much smaller
than its current class would suggest; this seems like a major problem!

\medskip\textbf{A Deeper Look at the Technical Challenges.}
To understand the difficulties, let us look closer at the proof
approaches. As is standard for flow time algorithms, we want to show
\emph{local competitiveness}; i.e., to prove that for any time $\ts$,
the number of active jobs $|A(\ts)|$ for the algorithm is at most some
factor $\rho$ times the number $|\OPT(\ts)|$ of active jobs for the
optimal algorithm. (Recall that our goal is to show a competitive ratio of
$\rho = O( \frac{\mu_1\mu_2}{\eps})$.)

We want to use dual-fitting 
on the \emph{excess}-based LP of \cite{bansal2014geometry}. For a
subset of jobs $S \subseteq J$ that arrives before time $\ts$, define
the \emph{excess} of a set $S$ of jobs as
$\ex(S) := \max\{0,(\sum_{j \in S} p_{j}) - (\ts - r_S)\}$ where
$r_S := \min_{j \in S} r_j$.  In words: even if an algorithm only
worked on jobs in $S$ since the first job in $S$ was released, the
remaining volume at time $\ts$ would be at least $\ex(S)$. The
constraints  of the excess-based LP say: the jobs in $S$ active at time $\ts$ must have total
size at least $\ex(S)$, which is clearly valid.
The core of the analysis is to find a feasible dual solution of value
$\Omega(\rho^{-1} \cdot |A(\ts)|)$.

The dual-fitting approach from \cite{GuptaKLSY25} for the clairvoyant case $(\eps = 1)$ works as
follows. Fix a class $k$, consisting of the jobs with estimate $\hat p_j \simeq 2^k$. Let $t_{\ge k}$ be the largest time such
that the algorithm processes a job of current class $\ge k$ during
$[t_{\ge k},t_{\ge k}+1]$. Now let $S_k$ be the set of jobs released
during $[t_{\ge k}+1,\ts]$. In the clairvoyant case, fresh jobs have
never been processed, since their size was known right at the
beginning; hence they are \emph{full}. 
The proof finds a \emph{partition of full jobs at time $\ts$}, denoted
by $F_0,F_1,\ldots,F_K$, and shows that if this partition has the
following properties:
\begin{enumerate} [label=(P\arabic*)]  
    \item \label{prop:1} we have a ``large'' excess: $\ex(S_{k}) \geq \sum_{k'<k} \sum_{j \in F_{k'}} p_j(\ts)$,
    \item \label{prop:2} if $j \in S_k$, then $p_j$ is ``not too large''
      (ideally, $p_j \leq 2^k$), and  \item \label{prop:3} if $j \in F_k$, then its remaining volume at time $\ts$ is ``large'' (ideally, $p_j(\ts) \geq 2^k$).
\end{enumerate}
then this implies a feasible dual solution with objective
value~$\Omega(\rho^{-1} \cdot |A(\ts)|)$ where $\eps = 1$ in this case. (We have given slightly
simplified variants of the actual properties, for ease of discussion.)

In general, it is unclear how to construct a partition that satisfies
these properties, especially for property
\ref{prop:1}. The authors in \cite{GuptaKPW26} introduce a ``reduced instance''
trick: they analyze a different instance obtained by reducing the
processing time and the class of each job, which can only improve the optimal schedule. The core idea in their construction is the following: if we reduce the
processing time and class of a full job in such a way that the
algorithm still does not touch the job, then the schedule of the
algorithm will not change at time $\ts$.
By exploiting the additional structure in the reduced instance, they
show that there is a partition of \emph{full jobs} satisfying Properties
\ref{prop:1}, \ref{prop:2}, and \ref{prop:3}.

\medskip\textbf{The Problem with Online Estimates.} 
For the $(\eps,\mu_1,\mu_2)$-online estimates, we use the same definition of $S_k$ as above. However, the notion of full jobs does not apply, as we process fresh jobs while their estimates have not yet arrived. A natural first candidate is to partition the set of fresh jobs (instead of full jobs) at time $\ts$ into $F_0,\ldots, F_K$ in the same way as in~\cite{GuptaKPW26}.
However, the reduced-instance trick is insufficient for achieving the three
properties above. The main culprits are (1) there may be no full jobs
(every job may receive some processing time because estimates need to be explored)
and (2) we can no longer show that the algorithm at time $\ts$ will
produce the same state as the original instance. To elaborate (2), the main challenge is
the following: a job $j$ that is fresh at $\ts$ becomes fresh for the final time before $\ts$ at some time $\fint_j$. The job receives $e_{j}(\fint_j)$ units of processing before time $\fint_j$ and does not receive any processing during $[\fint_j,\ts]$. Hence, when reducing the instance, we have to ensure that $j$ does not receive processing during $[\fint_j,\ts]$ \emph{and} that $j$ receives $e_{j}(\fint_j)$ units of processing before time $\fint_j$. The latter is not necessary when $\eps =1$ as $\fint_j = r_j$, and it turns out to be significantly more challenging.  In particular, it remains unclear whether a useful reduced instance exists on which the algorithm's schedule remains the same.

\medskip\textbf{Our Approach: Dual-fitting on an ``Effective'' LP.} To deal with these issues, 
we consider modified processing times $\tilde p_j$ 
satisfying $e_j(\ts) \leq \tilde p_j \leq p_j$, where $e_j(\ts)$ is the elapsed time of job $j$ at time $\ts$ in the algorithm's original schedule. To avoid confusion with the reduced instance, we refer to them as \emph{effective processing times}. The modified instance is called an \emph{effective instance} (given by $\tilde p_j$). 
We then dual-fit an \emph{effective LP} obtained 
from the original LP after replacing processing times $p_j$ with effective processing times $\tilde p_j$. Since $e_j(\ts) \leq \tilde p_j \leq p_j$, the effective LP value remains a lower bound of $|\opt(\ts)|$ (c.f. \Cref{lem:relaxation} for details). 
In this way, the algorithm's behavior on the instance arising from $\tilde p_j$ is irrelevant to the analysis since we perform dual fitting on an alternative LP relaxation. In particular, this allows us to bypass the previous barrier of using ``reduced instance" trick: we no longer need to prove that the algorithm does not change its schedule at time $\ts$. 
The next step is to determine $\tilde p_j$ 
such that we can find a partition of fresh jobs at time $\ts$ that satisfies Properties \ref{prop:1}, \ref{prop:2}, and \ref{prop:3} with respect to the same sets $S_k$ as in the original job instance.
Compared to the reduced instance in \cite{GuptaKPW26} when $\eps = 1$, our situation is now much more challenging because we cannot use properties of the algorithm's schedule for the effective instance.

The crux of our analysis is the new concept of \emph{effective classes}. For each job $j$, we carefully reduce $j$'s class to a smaller class $\redk_j$ which represents an \emph{effective class} of $j$.  We then partition the fresh jobs at time $\ts$ into these effective classes, yielding~\ref{prop:1}.
In fact, we show the even stronger property that \[
\tilde \ex(S_{k}) \geq \sum_{k'<k} \sum_{j \in F_{k'}} \tilde p_j(\ts)
\]  
for \emph{all} effective instances where $\tilde \ex(S_{k})$ is the excess defined w.r.t.\ the $\tilde p_j$'s. This is independent of whether the algorithm actually processes the instance defined by the $\tilde p_j$ or not, and gives us useful flexibility for the rest of the analysis. This stronger variant of~\ref{prop:1} now allows us to  freely choose an effective LP (as long as $e_j(\ts) \leq \tilde p_j \leq p_j$) such that the partition of fresh jobs into effective classes satisfies ~\ref{prop:2} and~\ref{prop:3} in addition to \ref{prop:1}.  The main challenge in our proof is finding the correct effective classes and proving that we indeed achieve the stronger variant of~\ref{prop:1}.

A side effect of our choice of effective processing times $\tilde p_j$ is that the effective class $\redk_j$ is potentially inconsistent with $\tilde p_j$ (ideally, we would like, for each job $j$, the effective processing time of $j$ to be a power of two of its effective class, but we may get a very tiny effective class but a large effective processing time). 
To handle the inconsistency between effective class and effective processing time, we develop a \emph{pruning} procedure that carefully removes selected fresh jobs from our partition that prevent us from achieving~\ref{prop:1}. We show that there is a limited number of such jobs, which allows us to still prove local competitiveness. 
Interestingly, the problematic jobs are those that change from ongoing to fresh again; a situation that never happens when $\eps = 1$. 
In a similar vein, the proofs of our excess bounds require a careful analysis of jobs that dynamically change their classes.

\subsection{Related Work}
\label{sec:related-work}

Becchetti and Leonardi~\cite{BecchettiL04} show that a randomized
variant of MLF, which ``smoothes'' the queue promotions using
exponentially distributed offsets \cite{KalyanasundaramP03}, achieves the optimal
$\Theta(\log n)$ competitive ratio for randomized nonclairvoyant
algorithms minimizing total flow~\cite{MotwaniPT94}.  
Motwani,
Phillips, and Torng~\cite{MotwaniPT94} show that every deterministic
nonclairvoyant algorithm has a competitive ratio of at least
$P = (\max_j p_j)/(\min_j p_j)$, which can be matched by any
nonpreemptive schedule.  For $m$ parallel identical machines, SRPT
achieves the best-possible competitive ratio of
$O(\log (\min\{\nf{n}{m}, P\}))$~\cite{LeonardiR97,LeonardiR07}, where
the best-known nonclairvoyant algorithm loses an additional
$O(\log n)$ factor~\cite{BecchettiL04}.  When all jobs are available
at time $0$, (a generalization of) Round-Robin achieves the optimal
nonclairvoyant competitive ratio of $2$ even for uniformly related
machines or identical machines with restricted assignment~\cite{MotwaniPT94,JagerLM26}.

Minimizing total flow has also been studied under \emph{speed
  augmentation}, where an algorithm's processor is $1+\eps$ times
faster than the optimum's.  Kalyanasundaram and
Pruhs~\cite{KalyanasundaramP00} showed that SETF is
$O(\nicefrac{1}{\eps})$-competitive for total flow, and there exist
$O_\eps(1)$-competitive nonclairvoyant algorithms for the weighted
objective~\cite{BansalD07} and on unrelated machines~\cite{ImKMP14},
for which no algorithm can achieve a constant competitive ratio even
with clairvoyance~\cite{BansalC09,LeonardiR07}.

The \emph{semi-clairvoyant} model of Bender, Muthukrishnan, and
Rajaraman~\cite{BenderMR04} assumes that the ``class'', defined as
$\floor{\log_\mu p_j}$, of each job $j$ is revealed when $j$ is
released; here $\mu$ is the granularity of prediction. The
\emph{predicted processing time} (also called \emph{robust flow time
  scheduling}) model of Azar, Leonardi, and
Touitou~\cite{AzarLT21,AzarLT22} assumes that each job $j$ comes with
a prediction~$\hat{p}_j$ for its true processing time $p_j$, such that
$\hat{p}_j/p_j \in [\nf{1}{\mu_2},\mu_1]$ for some distortion
parameters $\mu_1,\mu_2 \ge 1$. (The parameters may be known to the
algorithm; or it may be distortion-oblivious.) These predictions can
be motivated by, e.g., learning-based techniques for datacenter
scheduling~\cite{0001HLL0GGQZL24}. An optimal distortion-oblivious
$O(\mu_1\mu_2)$-competitiveness is now known \cite{GuptaKPW26}. 

Yingchareonthawornchai and Torng~\cite{YingchareonthawornchaiT17}
introduced the \emph{$\eps$-clairvoyant model}, motivated by, e.g.,
job profilers~\cite{hilman2018task,xie2021two,weerasiri2017taxonomy}
``learning'' the job sizes while processing jobs, or a GenAI server
learning about session lengths based on interactions thus far.
\cite{GuptaKLSY25} give an optimal
$\lceil \nf{1}{\eps} \rceil$-competitive algorithm for minimizing
total flow in this model.  However, it requires that job sizes are
given exactly when the remaining processing time hits the
$\eps$-percentile: this work does not address Assumption~\ref{item:1}.
When all jobs are available at time $0$, one can achieve competitive
ratios below the optimal nonclairvoyant ratio of $2$ even if the
signal timing is slightly incorrect~\cite{BenomarCLS25}.

Recently, \cite{Lindermayr26} considered a ``precedence-chains''-style
setting, where each job is composed of at most $m$
\emph{operations}. The processing time of the first
operation is revealed at job arrival, and the processing time of the $i$th operation is revealed once the $(i-1)$th operation is completed.
They give an
$O(m^2)$-competitive algorithm for total (job) flow, which builds
on~\cite{GuptaKPW26}. They need to handle \enquote{changing} job
classes as more operations are revealed, which is similar to our
setting. However, since the number of operations
$m$ bounds the number of class changes, they can just introduce
artificial jobs (called \emph{chunks}) for each class change, and lose
a factor
$m$ in the analysis.  In our setting, we can have $\Omega(\log
p_{\max})$ chunks per job, which cannot be bounded by
$\eps$ and
$\mu$. Thus, following the ideas of \cite{Lindermayr26} would lose a
factor of at least $\log p_{\max}$ for online
predictions. 

More broadly, our model falls into the area of learning-augmented
algorithms (aka algorithms with predictions), where algorithms
are equipped with an additional but potentially imprecise input to
overcome pessimistic and notorious worst-case lower
bounds~\cite{MitzenmacherV22}. Online scheduling problems are at the 
core of this body of literature, e.g.,
\cite{PurohitSK18,WeiZ20,ImKQP23,BampisDKLP22,EliasKMM24,BenomarP23,LindermayrM25permutation,AzarLT21,AzarLT22,ZhaoLZ22,AzarPT22,GuptaKPW26,BamasMRS20,AntoniadisGS22}.
As discussed above, predictions are usually given before the online
instance starts or are attached to arriving requests.  Our model can
thus be viewed as ``adversarially arriving online predictions'', which
we believe can be motivated for and applied to other online problems,
leading to a novel direction in the area.

 \subsection{Preliminaries and Notation}
\label{sec:preliminaries}

In our problem, jobs arrive over time: job $j$ arrives at time
$r_j \in \ZZ_{\ge0}$, and requires $p_j \in \ZZ_{\ge1}$ units of processing. 
A feasible schedule, during each time interval $[t,t+1]$ for
$t \in \ZZ_{\ge0}$, runs some job $j(t)$, as long as the job has already been
released, i.e., if $r_{j(t)} \leq t$. (Note that preemption is
allowed, so that job $j$ can be run in nonconsecutive time
intervals.) Let the \emph{elapsed time} for job $j$ at time
$t \in \ZZ_{\ge0}$ with $t \geq r_j$ be
$e_j(t) := \sum_{t'=0}^{t-1} \mathbf1[j = j(t')]$. Then the remaining
processing time $p_j(t) = p_j - e_j(t)$, and the job is said to be
complete when $p_j(t) = 0$.  The \emph{completion time} is
$C_j := \min\{t \mid p_j(t) = 0\}$, and the flow time is
$F_j := C_j - r_j$. Our objective is the total flow time $\sum_j F_j$.

\begin{defn}[Online Estimate Model]
  In the \emph{$(\varepsilon,\mu_1,\mu_2)$-online estimate model}
  for online scheduling, for each job $j$ the adversary chooses not
  only the release time $r_j$ and processing time $p_j$, but also
  (i)~an \emph{estimate} $\hat{p}_j \in \ZZ_{\ge1}$  satisfying
  $\nicefrac{p_j}{\mu_2}\leq \hat p_j \leq \mu_1 p_j$, and (ii)~an
  \emph{online estimate threshold} $q_j \in \ZZ_{\ge 0}$ that satisfies $q_j \leq (1-\eps) p_j$.
When the job arrives at time $r_j$, the algorithm does not know any
  other job parameters. At the first time $t$ when the algorithm has
  worked for $q_j$ units on job $j$, it is told~$\hat{p}_j$. 
  Finally, the algorithm knows when the job finishes, and
  so it learns the true $p_j$ at time $C_j$.
\end{defn}

We call the (integer) time $t$ at which the job receives its estimate $\hat{p}_j$
to be the \emph{signal time} $s_j$; note that
$p_j(s_j) \geq \eps p_j$. The job is said to be \emph{unknown} before
that time, and is \emph{known} from $s_j$ onwards.

\subsection{Paper Organization}
The rest of the paper is organized as follows. We define Balanced MLF precisely in \Cref{sec:alg}. In \Cref{sec:effective}, we describe our effective instance and its important properties. We prove the excess lemmas in \Cref{sec:the excess lemma}, and then carry out the dual-fitting argument in \Cref{sec:dual fitting}. Finally, we prove the matching lower bound in \Cref{sec:lower-bounds}.

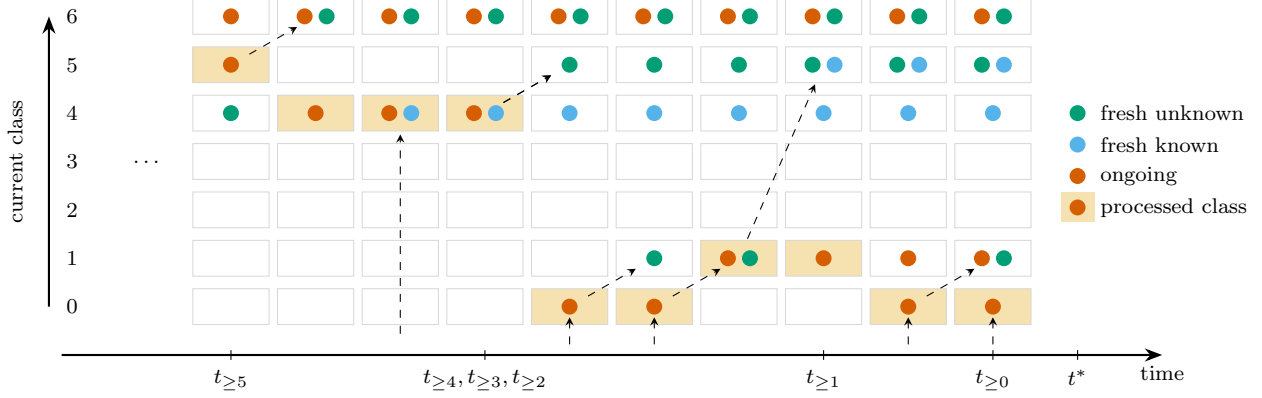
\begin{figure}
	\centering
	
	\begin{tikzpicture}[
  x=1.12cm,y=0.64cm,
]

\foreach \x/\k in {0/5,1/4,2/4,3/4,4/0,5/0,6/1,7/1,8/0,9/0}{
  \fill[algProcCell] (\x-.45,\k-.37) rectangle (\x+.45,\k+.37);
}

\foreach \x in {0,1,2,3,4,5,6,7,8,9}{
  \foreach \k in {0,...,6}{\draw[algCell] (\x-.45,\k-.37) rectangle (\x+.45,\k+.37);}
}
\node at (-1,3) {\scriptsize $\ldots$};

\foreach \k in {0,...,6}{\node[left=6pt] at (-1.50,\k) {\scriptsize $\k$};}
\draw[-{Stealth}, thick] (-2.15,0) -- (-2.15,6);
\node[rotate=90] at (-2.55,3) {\scriptsize current class};

\algOngoingDot{0}{5}{0}{0}
\algOngoingDot{1}{4}{0}{0}
\algOngoingDot{2}{4}{-.13}{0}
\algOngoingDot{3}{4}{-.13}{0}
\algOngoingDot{4}{0}{0}{0}
\algOngoingDot{5}{0}{0}{0}
\algOngoingDot{6}{1}{-.13}{0}
\algOngoingDot{7}{1}{0}{0}
\algOngoingDot{8}{0}{0}{0}
\algOngoingDot{9}{0}{0}{0}

\foreach \x/\k/\dx/\dy in {
  0/6/0/0, 
  1/6/-.13/0,
  2/6/-.13/0,
  3/6/-.13/0,
  4/6/-.13/0,
  5/6/-.13/0,
  6/6/-.13/0,
  7/6/-.13/0,
  8/6/-.13/0,
  9/6/-.13/0,
  8/1/0/0,
  9/1/-.13/0
}{\algOngoingDot{\x}{\k}{\dx}{\dy}}

\foreach \x/\k/\dx/\dy in { 
  0/4/0/0,
  1/6/.13/0,
  2/6/.13/0,
  3/6/.13/0,
  4/6/.13/0,
  4/5/0/0,
  5/1/0/0,
  5/5/0/0,
  5/6/.13/0,
  6/1/.13/0,  
  6/5/0/0,
  6/6/.13/0, 
  7/5/-.13/0, 
  7/6/.13/0,
  8/5/-.13/0,
  8/6/.13/0,
  9/5/-.13/0,
  9/6/.13/0,
  9/1/.13/0
}{\algFreshDot{\x}{\k}{\dx}{\dy}}

\foreach \x/\k/\dx/\dy in { 
  2/4/.13/0,
  3/4/.13/0,
  4/4/0/0,
  7/5/.13/0,
  5/4/0/0,
  6/4/0/0,
  7/4/0/0,
  8/4/0/0,
  8/5/.13/0,
  9/4/0/0,
  9/5/.13/0
}{\algFreshKnownDot{\x}{\k}{\dx}{\dy}}

\draw[-{Stealth}, thick] (-2,-1) -- (11,-1) node[below,font=\scriptsize] {time};

\draw[dashed,->,>=stealth,shorten <=8pt,shorten >=8pt] (0,5) -- (1,6);
\draw[dashed,->,>=stealth,shorten <=8pt,shorten >=8pt] (3,4) -- (4,5);
\draw[dashed,->,>=stealth,shorten <=8pt,shorten >=8pt] (3,4) -- (4,5);
\draw[dashed,->,>=stealth,shorten <=8pt,shorten >=8pt] (4,0) -- (5,1);
\draw[dashed,->,>=stealth,shorten <=8pt,shorten >=8pt] (5,0) -- (6,1);
\draw[dashed,->,>=stealth,shorten <=8pt,shorten >=8pt] (8,0) -- (9,1);
\draw[dashed,->,>=stealth,shorten <=8pt,shorten >=8pt] (6,1) -- (7,5);
\draw[dashed,->,>=stealth,shorten <=8pt,shorten >=8pt] (2,-1) -- (2,4);
\draw[dashed,->,>=stealth,shorten <=4pt,shorten >=4pt] (4,-1) -- (4,0);
\draw[dashed,->,>=stealth,shorten <=4pt,shorten >=4pt] (5,-1) -- (5,0);
\draw[dashed,->,>=stealth,shorten <=4pt,shorten >=4pt] (8,-1) -- (8,0);
\draw[dashed,->,>=stealth,shorten <=4pt,shorten >=4pt] (9,-1) -- (9,0);

\begin{scope}[shift={(0,-1)}]
\draw[] (0,0.1) -- (0,-0.1) node[below,font=\scriptsize] {$t_{\ge 5}$};
\draw[] (3,0.1) -- (3,-0.1) node[below,font=\scriptsize] {$t_{\ge 4},t_{\ge 3},t_{\ge 2}$};
\draw[] (7,0.1) -- (7,-0.1) node[below,font=\scriptsize] {$t_{\ge 1}$};
\draw[] (9,0.1) -- (9,-0.1) node[below,font=\scriptsize] {$t_{\ge 0}$};
\draw[] (10,0.1) -- (10,-0.1) node[below,font=\scriptsize] {$\ts$};
\end{scope}

\begin{scope}[shift={(10,4)}]
  \fill[freshgreen] (0,0) circle[radius=3pt];
  \node[anchor=west] at (.15,0) {\scriptsize fresh unknown};
  \fill[freshknown] (0,-.65) circle[radius=3pt];
  \node[anchor=west] at (0.15,-.65) {\scriptsize fresh known};
  \fill[ongoingred] (0,-0.65*2) circle[radius=3pt];
  \node[anchor=west] at (0.15,-0.65*2) {\scriptsize ongoing};
  \fill[algProcCell] (-.2,-0.65*3-.3) rectangle (.2,-0.65*3+.3);
  \fill[ongoingred] (0,-0.65*3) circle[radius=3pt];
  \node[anchor=west] at (0.15,-0.65*3) {\scriptsize processed class};
\end{scope}
\end{tikzpicture}
\caption{Exemplary execution of Balanced MLF. The rectangles indicate job queues (classes) at integer times after the balancing phases. The arrows indicate job promotions and arrivals. The times $t_{\ge k}$ are defined as in \Cref{def:tk-times}. Note that during the first displayed time, Balanced MLF does not make the job of smallest class ongoing because there are too few fresh jobs (cf.\ \Cref{lem:non-greedy}).}
\label{fig:alg-example}
\end{figure}

\section{Balanced MLF}
\label{sec:alg}

We start by defining the \emph{current class} $k_j(t) \in \{0,1,2,\ldots\}$ of a job $j$ at time $t$:
\[
    k_j(t) := \begin{cases}
        \lceil \log_2(e_j(t)+1) \rceil & \text{ if } t \leq s_j-1 \\
        \max\{ \lceil \log_2(e_j(s_j)+1) \rceil, 1+\lfloor \log_2 \hat p_j \rfloor \} & \text{ if } t \geq s_j .
    \end{cases}
\]
In words, until the estimate arrives, the current class of a job $j$
models the power-of-2 class of the elapsed time, with completely fresh
jobs having class $0$. Once the estimate arrives (at time $s_j$),
the current class of $j$ may increase if the predicted processing time
has a higher class than the current elapsed time.  Importantly, the
current class of a job is monotone nondecreasing over time, and every
job starts with current class $0$ unless $s_j=0$.  Since at time $s_j$
we have already verified that job $j$ has length at least $e_j(s_j)$,
we can assume w.l.o.g.\ that $\hat p_j \geq e_j(s_j)$, which in turn
ensures that $k_j(t) = 1+\lfloor \log_2 \hat p_j \rfloor$ if
$t \geq s_j$. 

\paragraph{Algorithm.} We next give a formal description of our algorithm. We assume w.l.o.g.\
that the schedule we produce has no idle periods, since we can perform
the analysis on each busy period separately. We also use the convention that $\min \emptyset = \infty$. 

At any integer time $t \geq 0$, we partition the set of active jobs
$A(t)$ into \emph{ongoing} jobs $\apart(t)$ and \emph{fresh} jobs
$\afull(t)$. To do so, we inductively maintain another partition of $A(t)$ into auxiliary sets $\capart(t)$ and $\cafull(t)$ as follows. Given
the sets $\capart$ and $\cafull$ from iteration $t-1$ (where
$\capart(-1) = \cafull(-1) = \emptyset$), and a set $N$ of new jobs
arriving at time $t$, we define
$\apart(t), \afull(t)$ via the \emph{balancing} phase. This is followed by a \emph{processing} phase, which defines  $\capart(t), \cafull(t)$. These two phases are shown in \Cref{alg:balance,alg:process}. 

For the balancing phase, we use the auxiliary sets $\capart(t - 1)$
and $\cafull(t-1)$ and treat a set of new jobs $N$ as fresh jobs. If
the smallest class of any fresh job is smaller than the smallest class
of all ongoing jobs, and the fresh jobs are at least a quarter of the
total number of active jobs, we change the fresh job with smallest
class to be ongoing. For the processing phase, we process the ongoing job with the smallest class; if its class increases, then make it fresh. 
The following \Cref{obs:1} immediately follows from the definition of the algorithm.

\begin{algorithm}[H]
  \caption{Balancing Phase at time $t$}
  \label{alg:balance}
  let $\mathrm{On}(t) \gets \capart(t-1)$, $\mathrm{Fr}(t)
  \gets \cafull(t-1) \cup N$ \tcp*{local variables}
  let $m :=
  |\mathrm{On}(t)|+|\mathrm{Fr}(t)|.$ \;
      
  \If{$\min \{k_j(t) \mid j \in \mathrm{Fr}(t)\} < \min\{k_j(t)
    \mid j \in \mathrm{On}(t)\}$ and $|\mathrm{Fr}(t)| \geq \frac14 m$}{
    let $j \in \arg \min \{k_{j'}(t) \mid j' \in \mathrm{Fr}(t)\}$ \;
    update
    $\mathrm{Fr}(t)  \gets \mathrm{Fr}(t) \setminus \{j\}$ and
    $\mathrm{On}(t) \gets \mathrm{On}(t) \cup \{j\}$
  }
  let $\aon(t) \gets \mathrm{On}(t)$ and $\afull(t) \gets \mathrm{Fr}(t)$.  \tcp*{Non-mutable variables}
\end{algorithm}

\begin{algorithm}[H]
  \caption{Processing Phase during $[t,t+1]$}
   \label{alg:process}
  let $j(t) \gets \arg\min \{k_{j'}(t) \mid j' \in \aon(t)\}$
  \tcp*{$j(t)$ is 
    ongoing job with smallest current class}
  let $k(t) \gets k_{j(t)}(t)$ \tcp*{class of job processed during
    $[t,t+1]$}
  process job $j(t)$ during $[t,t+1]$ \;
  \If{\normalfont{current} class of $j(t)$ increases, i.e., $k_{j(t)}(t+1) > k_{j(t)}(t)$}{
    $\cafull(t) \gets \afull(t) \cup \{j(t)\}$ and $\caon(t) \gets
    \aon(t) \setminus \{j(t)\}$.     \tcp*{refresh $j(t)$}
  }
  \ElseIf{$j(t)$ \normalfont{is} completed}{
    $\cafull(t) \gets \afull(t)$ and $\caon(t) \gets
    \aon(t) \setminus \{j(t)\}$.     \tcp*{remove $j(t)$}
  }
  \Else {
    $\cafull(t) \gets \afull(t)$ and $\caon(t) \gets
    \aon(t)$.    
  }
\end{algorithm}

\begin{observation}
\label{obs:1}
  At time $t$,
  \begin{align}\label{eq:alg cond}
    \min_{j \in \afull(t)} k_j(t) \geq  \min_{j \in \aon(t)} k_j(t)
    \quad \textbf{ or } \quad |\afull(t)| < \frac14 |A(t)| \ .
  \end{align}
\end{observation}
\begin{proof}
    Consider the condition of the if-statement in \Cref{alg:balance}. If it does not hold, then we are done. If it holds, then after the update we have  $\min_{j \in \afull(t)} k_j(t) \geq  \min_{j \in \aon(t)} k_j(t)$.
\end{proof}

\subsection{Properties of Balanced MLF}

We first study useful properties of Balanced MLF. 
All proofs from this section are deferred to \Cref{sec:proofs}.
First, the algorithm description implies that at any time, there is at most one ongoing job within each current class.

\begin{restatable}[Ongoing current classes]{observation}{stackmonotone}
\label{lem:stack-monotone}
At every time $t$, 
for each class $k$, there is at most one ongoing job $j \in \aon(t)$
with $k=k_j(t)$.

\end{restatable}

We next show that our balancing rule implies the ``special rule'' property of \cite{BecchettiLMP04}. Intuitively, it says that if we process a job of current class $k$ at time $t$, then there can exist at most one fresh job of smaller current class at this time.
While this property was mainly used in the context of 
static advice (using full and partial jobs 
instead of fresh and ongoing jobs~\cite{GuptaKPW26}), we show that 
this also holds for our more general algorithm.

\begin{restatable}{lemma}{nongreedy}\label{lem:non-greedy}
  Suppose a job $j$ of class $k_j(t)$ is processed during the timestep $[t,t+1]$. If there is a job $j_1 \in \afull(t)$ with $k_{j_1}(t) < k_j(t)$, then there cannot exist a job $i \in A(t) \setminus \{j,j_1\}$ with $k_i(t) \le k_j(t)$.
\end{restatable}

Finally, we show that the balancing rule of Balanced MLF ensures that
we never make too many jobs ongoing, hence always keep a constant fraction 
of active jobs fresh.
This allows us to focus on fresh jobs later in the analysis.

\begin{restatable}[Many Fresh Jobs]{lemma}{manyfresh}\label{lem:many-fresh}
At every time $t$, it holds that
\(
    |\afull(t)| \ge \frac14 |A(t)| - 1.
\)
\end{restatable}

\section{The Effective Instance}\label{sec:effective}

For the analysis, we consider a fixed instance 
and the state of Balanced MLF during its execution on that instance.
To show that Balanced MLF is $c$-competitive, we prove that at any time $t$, its number of active jobs $|A(t)|$
is at most $c$ times the number of active jobs $|\opt(t)|$ in an optimal solution at time $t$.
This means that Balanced MLF is \emph{locally $c$-competitive}. Since at every time each active job contributes one unit to its flow time, local competitiveness clearly implies competitiveness.

For the remaining analysis, fix a time $\ts$ for which we want to show local competitiveness.
Moreover, we have the following assumptions without loss of generality:
\begin{itemize}
    \item The algorithm (and thus the optimum) never idles before time $\ts$, and $|\opt(\ts)| \geq 1$ and $|A(\ts)| \geq 1$. This can be achieved by analyzing all busy phases separately.
    \item We only consider jobs that arrive before time $\ts$, that is, the instance is composed of jobs $J := \{j \mid r_j < \ts\}$. All jobs released strictly after $\ts$ do not influence the local competitive ratio at time $\ts$, and all jobs released exactly at $\ts$ can only improve the local competitive ratio at time $\ts$ because they are part of $A(\ts) \cap O(\ts)$. \end{itemize}

\subsection{Effective Processing Times and Effective Classes}

We now define an effective instance by  reducing the
processing times and classes of some jobs in order to obtain an
LP relaxation with more structural properties. In contrast to their
use in the prior work of \cite{GuptaKPW26}, we do not assume that the
algorithm processes the instance defined by the reduced processing
times, and we do not claim that the algorithm behaves the same on that
instance and on the original instance. Instead, we use these reduced
processing times and classes solely to lower bound $|\opt(\ts)|$.
Allowing this flexibility gives us more control on the kinds of
changes allowed, and seems crucial to our proof.

\begin{defn}[Final classes $\fink_j$ and $\fint_j$]  \label{def:final-class}
    For a job $j$, define
    \[ 
    \fink_j := \begin{cases}
        k_j(\ts) \qquad & j \in A(\ts)\\
        k_j(C_j-1) \qquad & \text{otherwise}
    \end{cases}
    \]
    as the \emph{final class} of $j$, where $C_j$ is the completion time of $j$. We use $\fint_j$ to denote the time at which $j$ was promoted to its final class, i.e., the earliest time $t$ where $k_j(t) = \fink_j$. 
    Note that we might have $r_j = \fint_j$ if $q_j = 0$, i.e., the estimate is available upon arrival of job $j$.
\end{defn}

We assume w.l.o.g.\ that $\fint_j < \ts$ for all $j \in \afull(\ts)$: If a job $j$ becomes fresh at exactly $\ts$, i.e., $\fint_j = \ts$, then it must be processed during $[\ts -1,\ts]$ and we can just continue to treat it as ongoing at time $\ts$. Since $j$ is processed during $[\ts-1,\ts]$ and we only consider the schedule until time $\ts$, this does not affect the schedule at all.

\begin{lemma} \label{lem:fresh tfj}
    If $j \in \afull(t^*)$, then  $j$ is not processed during $[\fint_j, t^* ]$. In particular, $j \in \afull(t)$ for every $\fint_j \le t \le \ts$.
\end{lemma}
\begin{proof}
 Suppose $j$ is processed some time in $[\fint_j,\ts]$. By definition of $\fint_j$, $j$ was promoted and became fresh at $\fint_j$. The job $j$ must then become ongoing after $\fint_j$ and then become fresh again by time $\ts$, contradicting the definition of $\fint_j$. 
\end{proof}

A useful property of time $\fint_j$ is that a job $j \in \afull(\ts)$ either becomes known exactly at time $\fint_j$ or remains unknown until time $\ts$, as formalized by the following lemma.

\begin{lemma}
    \label{lem:final:class:pred}
    All jobs $j$ satisfy $s_j \ge \fint_j$. That is, the estimate of $j$ is not revealed before time $\fint_j$. If $j \in \afull(\ts)$, then either $s_j = \fint_j$ or $s_j \ge \ts$.
\end{lemma}

\begin{proof}
    For the first part of the lemma, assume for the sake of contradiction that $s_j < \fint_j$. By definition of current classes, the class of $j$ remains the same at any time $s_j \le t \le \ts-1$. However, this is a contradiction to $j$ reaching its final class at $\fint_j > s_j$.

    For the second part of the lemma, assume that $j \in \afull(\ts)$. Then, the two facts that $j$ is not touched during $[\fint_j,\ts]$ (cf.~\Cref{lem:fresh tfj}) and that $\fint_j \le s_j$ imply that $s_j = \fint_j$ or $s_j \ge \ts$.
\end{proof}

\begin{defn}[Effective classes $\redk_j$]
    \label{def:effective:class}
For every job $j$, we
define its \emph{effective class} as 
    \[
        \redk_j := \begin{cases}
        \fink_j \quad &\text{ if } j \not\in \afull(\ts), \text{ and } \\
      \min\{\fink_j, \max_{\fint_j\le t < \ts } k(t)\} &\text{ if } j \in \afull(\ts).
        \end{cases}
      \]
      where $k(t)$ is the current class of the job processed at time $t$.
\end{defn}
Note that in the definition above, we do not consider time $\ts$ in the maximum. This is because we are interested in local competitiveness at time $\ts$, and do not care which current class is processed during $[\ts,\ts+1]$.

We associate every class $k\ge 0$ with lower and upper bounds 
$\lambda_k := \max\{1,2^{k-1}\}$ and
$\Lambda_k := 2^k$.
Note that $\Lambda_k \le 2\lambda_k$ for every $k\ge 0$.
For our LP-relaxation, we will replace the processing time of a job $j$ with the \emph{effective processing time} as defined below. We select the effective processing time $\redp_j$ of a job $j \in \afull(\ts)$ in a way that only underestimates $p_j$ but still guarantees that the effective remaining time $\redp_j - e_j(\ts)$ is neither too small nor too large compared to the upper bound $\Lambda_{\redk_j}$ of the corresponding effective class. By~\Cref{lem:fresh tfj,lem:final:class:pred}, we have $p_j - e_j(\ts) \ge \eps p_j$ for all $j \in \afull(\ts)$, which motivates the following definition.

\begin{defn}[Effective (remaining) processing times $\redp_j$]
    \label{def:effective:time}

    For every job $j$, we define its \emph{effective processing time} as
    \[
        \redp_j := \begin{cases}      
             e_j(\ts) & \text{ if } j \notin \afull(\ts), \text{ and } \\
            e_j(\ts)+ \eps \cdot \min\{p_j, \Lambda_{\redk_j}\}  & \text{ if } j \in \afull(\ts) \ .
        \end{cases}
    \]
    We define $j$'s \emph{effective remaining time} at any time $t \leq \ts$ as $\redp_j(t) := \redp_j-e_j(t)$.
\end{defn}

\begin{remark}
    We remark that we might have $\redk_j \ll \fink_j$. Since $e_j(\ts)$ can be as large as $\lambda_{\fink_j}$ if $j \in \afull(\ts)$, it can happen that $\redp_j \gg \Lambda_{\redk_j}$ and even  $\redp_j \gg \mu_2 \Lambda_{\redk_j}$. In that sense, the effective processing time of a job $j \in \afull(\ts)$ can be inconsistent with its effective class $\redk_j$. Part of our proof is to show that this inconsistency between $\redp_j$ and $\redk_j$ only happens for a sufficiently small fraction of jobs in $\afull(\ts)$ that can be ignored at the cost of a constant factor in the local competitive ratio. 
\end{remark}

The next observation shows that the effective processing times only underestimate the actual processing times and that they are consistent with the elapsed times of the algorithm on the original instance. In particular, the observation implies that the effective remaining times are non-negative.

\begin{observation}
    \label{obs:relaxed:processing:times}
For each job $j$, it holds that $e_j(\ts) \le \redp_{j} \le p_j$.
\end{observation}

\begin{proof}
 We distinguish two cases.
    \begin{itemize}
        \item If $j \notin \afull(\ts)$, then $\redp_j =e_j(\ts) \le p_j$ as no job can receive more than $p_j$ amount of processing.
        \item If $j \in \afull(\ts)$, then, by \Cref{lem:fresh tfj}, $j$ is part of $\afull(t)$ for every $\fint_j \le t \le \ts$, and $j$ is not touched during $[\fint_j,\ts]$. At time $\fint_j$, job $j$ is either still unknown or has just become known (cf.~\Cref{lem:final:class:pred}). In either case, $p_j(\fint_j) \ge \eps p_j$ and thus $e_j(\ts) = e_j(\fint_j) \le (1-\eps) p_j$. We can conclude with
        \[
        e_j(\ts) \le \redp_j = e_j(\ts)+ \eps \cdot \min\{p_j, \Lambda_{\redk_j}\} \le (1-\eps) p_j + \eps p_j = p_j. \qedhere
        \]
    \end{itemize}
\end{proof}

As mentioned above, the effective remaining time of a job $j$ is chosen to be small compared to $\Lambda_{\redk_j}$. 

\begin{lemma}\label{lem:p-class-bounds}
    For each job $j$, it holds that $\redp_{j}(\ts) \leq \eps \cdot \Lambda_{\redk_j}$. 
\end{lemma}

\begin{proof}
Fix a job $j$ with $r_j \le \ts$. One of the following cases must hold.
\begin{itemize}
    \item If $j \in \apart(\ts)$, then $\redp_{j}(\ts) = e_j(\ts)-e_j(\ts)=0 \le  \eps \cdot \Lambda_{\redk_j}$.
     \item If $j \not\in A(\ts)$, then $\redp_{j}(\ts) = p_j(\ts) = 0\le  \eps \cdot \Lambda_{\redk_j}$. 
\item If $j \in \afull(\ts)$, then $j$ has not been touched after reaching its final class (cf.~\Cref{lem:fresh tfj}). That is, $j$ has not been touched during $[\fint_j,\ts]$. By~\Cref{def:effective:time}, the effective remaining time at $\ts$ is 
    $$\redp_j(\ts) = \redp_j(\fint_j) = \eps \cdot \min\{p_j, \Lambda_{\redk_j}\} \le \eps \Lambda_{\redk_j}.$$
\end{itemize}
This concludes the proof of the lemma.
\end{proof}

\subsection{The LP Relaxation}
\label{sec:lp-relaxation}

Based on the effective processing times, we define an LP relaxation to
lower bound $|\opt(\ts)|$.  This relaxation is based on the
knapsack-cover relaxation (e.g., given by \cite{bansal2014geometry})
applied to the instance with effective processing times $\redp_j$
instead of true processing times $p_j$. Intuitively, reducing processing times can only
weaken the relaxation.

For a subset of jobs $S \subseteq J$, define the \emph{effective excess} as $\redex(S) := \max\{0,(\sum_{j \in S} \redp_{j}) - (\ts - r_S)\}$ where $r_S :=  \min_{j \in S} r_j$. The LP relaxation for $|\opt(\ts)|$ can be written as follows: 
\begin{alignat*}{3}
(\lppart) \quad \min \; & \sum_{j \in J} x_{j} \\
    \text{s.t.} \;  & \sum_{j \in S} \min \{ \redp_{j}, \redex(S) \} \cdot x_{j} \geq \redex(S) &&\quad \forall S \subseteq J \\
    & x_{j} \geq 0 &&\quad \forall j \in J
\end{alignat*}
Its dual can be written as follows:
\begin{alignat*}{3}
(\dppart) \quad \max \;  &  \sum_{S \subseteq J} \redex(S) \cdot y_{S} \\
    \text{s.t.} \;  & \sum_{S: j \in S} \min \{ \redp_{j}, \redex(S) \} \cdot y_{S} \leq 1 &&\quad \forall j \in J \\
    & y_{S} \geq 0 &&\quad \forall S \subseteq J
\end{alignat*}

\begin{lemma}\label{lem:relaxation}
    The optimal objective value of $(\lppart)$ is at most $|\opt(\ts)|$.
\end{lemma}

\begin{proof}
    Consider an optimal schedule $Z_p$ for the instance with processing times $(p_j)_{j \in J}$, which has $|\opt(\ts)|$ active jobs at time $\ts$.
    Since $\redp_j \leq p_j$ for every job $j \in J$ by \Cref{obs:relaxed:processing:times}, 
    we can turn $Z_p$ into a feasible schedule $Z'_p$
    for the instance with processing times $(\redp_j)_{j \in J}$ by introducing an idle time whenever we process the last $p_j-\redp_j$ units of a job $j$. 
    Since each job completes in $Z'_p$ no later than in $Z_p$, there are at most $|\opt(\ts)|$ jobs active in $Z'_p$ at time $\ts$.
    Let $Z_{\redp}$ be an optimal
    schedule for the instance with processing times $(\redp_j)_{j \in J}$.
    Since $Z_{\redp}$ is also locally $1$-competitive~\cite{Schrage68}, 
    the existence of the feasible schedule $Z'_p$ implies that
    $Z_{\redp}$ has at most $|\opt(\ts)|$ active jobs at time $\ts$.
    \cite{GuptaKPW26} show that the optimal objective value of $(\lppart)$ 
    is at most the number of active jobs in $Z_{\redp}$ at time $\ts$.
    This implies the claim.
\end{proof}

\paragraph{Proof organization.}

In the rest of the proof, we show an approximately feasible dual
solution with large value, completing the proof of local competitiveness.
We organize the remaining proof into two sections. In \Cref{sec:the
  excess lemma}, we identify sets of jobs $S$ with large effective
excess $\redex(S)$, as the dual variables $y_S$ corresponding to those
sets can strongly contribute to the dual objective value. We identify such sets in \Cref{sec:consistent fresh jobs}, and
prove that they have large excess guarantees in \Cref{lem:strong and
  weak excess}.

Finally, in \Cref{sec:dual fitting}, we define the dual solution with
objective function value within a small error of the true
number of jobs, which also satisfies approximate feasibility (for
which we use the excess lemmas from \Cref{lem:strong and weak excess}).

\section{The Excess Lemmas} \label{sec:the excess lemma}

We next identify sets of jobs $S$ for which we later will define nontrivial duals $y_S$. We further prove lower bounds on their effective excess $\redex(S)$ (the excess lemmas), which we use in the dual fitting.

\subsection{Finding Consistent Fresh Jobs} \label{sec:consistent fresh jobs}

As outlined above, we aim at constructing a dual solution with objective value $\Omega(\frac{1}{\rho} |A(\ts)|)$. By \Cref{lem:many-fresh}, it suffices to achieve dual objective value  $\Omega(\frac{1}{\rho} |\afull(\ts)|)$. 
In order to construct such a dual solution, we partition the jobs $j \in \afull(\ts)$ into sets $F(k)$ according to their effective classes $\redk_j$. 
Afterwards, we use a pruning procedure on the sets $F(k)$ to construct sets $\bar{F}(k)$ by removing a small number of jobs. Before going into details, we outline the reason for using the pruning procedure.

An important part of our dual fitting analysis is to identify sets of jobs $S$ with large effective excess $\redex(S)$, as the dual variables $y_S$ corresponding to those sets can strongly contribute to the dual objective value. A natural approach to identify such sets is to find intervals $I = [t,\ts]$ such that the algorithm during $I$ only works on jobs that are released during $I$. If $S$ is the set of all jobs that are released during that interval, this then implies 
$$
 \redex(S) \ge \sum_{j \in S} \redp_j(\ts).
$$
In order to make sure that the RHS of this inequality is indeed \enquote{large}, we would like $S$ to contain many fresh jobs because the freshness of such jobs implies that they have a relatively large remaining processing time at time $\ts$, as we will formally prove in~\Cref{lem:full-lower-bounds} below.

To construct such sets $S$, we consider a chain of time intervals ending in $\ts$ during which only jobs of bounded (actual, not effective) class are processed.

\begin{defn} [Time $t_{\geq k}$]\label{def:tk-times}
For every class $k \ge 0$, we define $t_{\geq k}$ as the last time before $\ts$ when the processed job $j(t_{\geq k})$ has current class at least $k$, that is, $k(t_{\geq k}) \geq k$ and $k(t) < k$ for all $t \in [t_{\geq k} + 1, \ts]$. 
If this does not exist, that is, before time $\ts$ no job of class $\geq k$ is processed, we set $t_{\ge k} = 0$.
Since we assume that the algorithm does not idle before $\ts$, we have $t_{\ge 0}=\ts-1$.
Finally, we introduce the shorthand notation $t_{>k} := t_{\ge k+1}$ for each class $k$.
\end{defn}

For the most part, the choice of intervals $I = [t_{\ge k} + 1, \ts]$ guarantees that the fresh jobs $\bigcup_{k' < k} F(k')$ are indeed released during $I_k$, which, as outlined above, is useful for showing that the jobs released during $I_k$ have a large effective excess.
However, as we will later see, there can be jobs in $\bigcup_{k' < k} F(k')$ that are not released during $I_k$. 
In particular, this can be the case for the job that is processed during $[t_{\ge k}, t_{\ge k}+ 1]$ if it becomes fresh again at time $t_{\ge k}+ 1$. Note that such jobs are specific to $(\eps,\mu_1,\mu_2)$-online estimates. In the special case where estimates are available upon job release, these jobs cannot exist.
The pruning procedure that transforms the sets $F(k)$ into subsets $\bar{F}(k)$ will remove these exception jobs. Since we will show that $|\bigcup_k \bar{F}(k)| \in \Omega(|\bigcup_k F(k)|)$, removing these jobs and finding a dual solution of the order $\Omega(\frac{1}{\rho}|\bigcup_k \bar{F}(k)|)$ will still be strong enough. 

\begin{defn}[Partition of Fresh Jobs]
    For each class $k$, define
    $F(k) = \{ j \in \afull(\ts) \mid \redk_{j} = k\}$.
\end{defn}

\begin{observation}
    \label{obs:fresh:job:partition}
    $\bigcup_k F(k) = \afull(\ts)$.
\end{observation}

Our pruning procedure removes some special jobs from the sets $F(k)$ and obtains the sets $\bar{F}(k)$ as follows:
\begin{enumerate}
    \item For each $k$, initialize $\bar{F}(k)$ with $F(k)$.
    \item If it exists, let $k$ denote the largest effective class such that 
    \begin{enumerate}
        \item [(i)] $\bar{F}(k)$ is not empty,
        \item [(ii)] the job $j$ that is processed during $[t_{\ge k},t_{\ge k}+1]$ becomes fresh at $t_{\ge k}+1$ and is not processed again during $[t_{\ge k}+1,\ts]$ where  $[t_{\geq k},t_{\geq k} + 1]$ denotes the last time interval before $\ts$ when the algorithm processes a job of class $\ge k$, and \begin{itemize}
            \item  
            (The time $t_{\ge k}$ exists: If $t_{\ge k}$ does not exist, then $k(t) \le k-1$ for all $t \in [0,\ts - 1]$ and, thus, $\redk_j = \min\{\fink_j, \max_{\fint_j\le t < \ts } k(t)\} \le \max_{0 \le t < \ts} k(t) \le k-1$ for all $j \in \afull(\ts)$. However, then $\bar{F}(k) \subseteq F(k) = \emptyset$ and Condition (i) does not hold.)
        \end{itemize} 
        \item [(iii)] $j \in \bar{F}(\redk_{j})$, i.e., $j$ has not been removed. 
        Note that this condition is necessary for two reasons: to ensure termination by preventing the same $k$ from being considered multiple times in Step $2$~and to handle situations where two classes $k$ and $k'$ have $t_{\ge k} = t_{\ge k'}$.

    \end{enumerate}  

    \item Remove the job $j$ of the previous step from $\bar{F}(\redk_{j})$.
    \item Repeat from the second step, until no such index exists anymore.
\end{enumerate}

Since we only remove jobs from the partition, we have the following observation.

\begin{observation}
    \label{obs:bar F subseteq}
For each class $k$, we have
    $\bar F(k) \subseteq F(k).$
\end{observation}

We next show that we prune only a few jobs from the partition. 
The full proof is deferred to \Cref{sec:proofs}.
We remark that if $\eps=1$, the pruning procedure removes no jobs, since a job that is processed cannot become fresh again. Hence, $|\bigcup_k \bar{F}(k)| = |\bigcup_k F(k)| = |\afull(\ts)|$. 

\begin{restatable}{lemma}{manyinbarF}\label{lem:many-in-bar-F}
    We have $|\bigcup_k \bar{F}(k)| \ge \frac{1}{2} \cdot |\bigcup_k F(k)| = \frac{1}{2} \cdot |\afull(\ts)|$.
\end{restatable}

\begin{proof}[Proof sketch]
Each removed job is charged to a distinct class that triggered its removal. 
The key point is that when a class $k$ triggers the removal of a job $j$, 
the removed job has effective class strictly smaller than $k$. 
Thus, the nonempty set $\bar F(k)$ that witnessed the trigger contributes 
a job that is never removed by this or any later trigger. 
Hence, at least one job remains for every job removed.
\end{proof}

For the dual fitting, our goal is to construct a feasible dual solution with objective value $\Omega(\frac{1}{\rho} \cdot |\bigcup_k \bar{F}(k)|)$. Using,~\Cref{lem:many-in-bar-F},~\Cref{obs:fresh:job:partition}, and \Cref{lem:many-fresh}, this implies a feasible dual solution of value $\Omega(\frac{1}{\rho} \cdot |A(\ts)|)$.

\subsection{Strong and Weak Excess Lemmas} \label{lem:strong and weak excess}

We now identify sets of jobs $S$ and show that they have
large effective excess $\redex(S)$. For every integer time $t \in \{0,1,2,\ldots,\ts\}$, let $S(t) := \{j \mid r_j \in [t,\ts]\}$; in particular, $S(0) = J$.
The sets we use are nested; more precisely,
\begin{equation}
    S(t_{\ge 0}+1) 
    \subseteq S(t_{>0}+1) = S(t_{\ge 1}+1) 
    \subseteq S(t_{> 1}+1)
    \subseteq \ldots \subseteq S(0) = \{j \mid r_j < \ts \} = J \, . \label{eq:S-chain}
\end{equation}

For the remainder of this section, fix a class $k$. 

\subsubsection{The Strong Excess Lemma}

\begin{lemma} (Strong Excess Lemma)
\label{lem:strong-excess}
If $\bar{F}(k) \neq \emptyset$, then $\redex(S(t_{\geq k}+1)) \geq \sum_{k' < k} \sum_{j \in \bar{F}(k')} \redp_{j}(\ts)$. 
\end{lemma}

\begin{remark}
    We remark that the construction of the sets $\bar{F}$ is only used in the very last argument in the proof of~\Cref{lem:strong-excess}. Everything before that point would work the same way if we used $F$ instead of $\bar{F}$. The same holds for all the following parts of the analysis.
\end{remark}

Our proof of the strong excess lemma will rely on the following claim:

\begin{lemma}
    \label{claim:strong:excess}
    If $F(k) \neq \emptyset$, then for all jobs $j \in A(t_{\ge k})$ we have $k_j(t_{\ge k}) \ge k$. 
\end{lemma}

\begin{proof}
    By assumption that $F(k) \neq \emptyset$, there is at least one job $j \in \afull(\ts)$ with $\redk_j = k$.

    We first argue that $j \in \afull(t_{\ge k})$ and that $j$ is not processed during $[t_{\ge k},\ts]$, i.e., $j \in \afull(t)$ for all $t$ with $t_{\ge k}\le t \le \ts$. By \Cref{lem:fresh tfj}, it is enough to prove $\fint_j \le t_{\ge k}$. 
For the sake of contradiction, assume that $\fint_j \ge t_{\ge k} + 1$. Then, by~\Cref{def:effective:class}, we have \[ k =\redk_{j} \le \max_{\fint_j \le t < \ts} k(t) \le \max_{t_{\ge k} + 1 \le t < \ts} k(t) < k.\] Let $j'$ be the job processed during $[t_{\ge k}, t_{\geq k}+1]$. 
    Note that $j \not= j'$ since $j$ is not processed during $[t_{\ge k},\ts]$. By the definition of $t_{\geq k}$, the class of $j'$ is $k_{j'}(t_{\geq k}) \geq k$.

    For the remaining proof, we distinguish two main cases of the current class of $j$ at time $t_{\ge k}$.

    \textbf{Case 1:} 
    {\bf $k_j(t_{\geq k}) \leq k$}. 
    We remark that in this case we must have $k_j(t_{\geq k}) = k$. Otherwise, the effective class of $j$ would be $\redk_j < k$, which contradicts $\redk_j = k$.

We can apply \Cref{lem:non-greedy} to $j$ and $j'$ and time $t_{\geq k}$ 
    to conclude that there cannot exist any job $i \in \afull(t_{\geq k}) \setminus \{j,j'\}$ with $k_{i}(t_{\ge k}) < k$. More precisely, if there was such a job $i$, then we could use~\Cref{lem:non-greedy} on $i$ and $j'$ to conclude that $j$ does not exist; a contradiction.

    \textbf{Case 2:} 
     {\bf $k_j(t_{\geq k}) > k$.}    
For the sake of contradiction, assume that there exists a job $i$ with current class $k_i(t_{\ge k}) < k$. Note that $i \in \afull(t_{\geq k})$ because  $j'$ is the minimum class in the ongoing set at time $t_{\geq k}$. Next, observe that
    \begin{align} \label{eq:violation}
    \min_{j \in \afull(t)} k_j(t) < \min_{j \in \aon(t)} k_j(t)
    \end{align}
    because
    \begin{align*} 
       \min_{j \in \afull(t)} k_j(t)   \leq  k_i(t_{\ge k}) < k \leq k_{j'}(t_{\geq k}) = \min_{j \in \aon(t)} k_j(t)  .
    \end{align*}  

    It is sufficient to prove that
    \begin{equation}    
    |\afull(t_{\ge k})| \ge \frac{1}{4} |A(t_{\ge k})| .
    \label{eq:claim:3}
    \end{equation}
    This is because \Cref{eq:violation,eq:claim:3} would contradict \Cref{eq:alg cond} at time $t_{\geq k}$. Intuitively, \Cref{eq:violation,eq:claim:3} imply that the algorithm should have moved $i$ to the ongoing set during the balancing phase at time $t_{\geq k}$ which contradicts $i$ being fresh at that time.

    To prove~\Cref{eq:claim:3}, we split into two subcases.
    
    \textbf{Case 2.1:}
    If there is no $k' > k$ such that $t_{\ge k'} \not= t_{\ge k}$, then all ongoing jobs at time $t_{\ge k}$ have class $\le k$. Since $j'$ is processed and has class $\ge k$, this implies that $|\apart(t_{\ge k})| = 1$ because of \Cref{lem:stack-monotone}. This also implies that the jobs $i$ and $j$ are fresh at time $t_{\ge k}$. Thus, 
    $$|\afull(t_{\ge k})| = |A(t_{\ge k})| - 1 \ge \frac{1}{4} |A(t_{\ge k})|.$$

    \textbf{Case 2.2:}
    Otherwise, let $k'$ be the smallest class with $k' > k$ such that $t_{\ge k'} \not= t_{\ge k}$. 
    Since the effective class of $j$ is $k$ and $\fink_j 
    \geq k_j(t_{\geq k}) > k$ by the assumption of Case 2, we have $\fint_j \in [t_{\ge k'}+1,t_{\ge k}]$.
    Moreover, we have $k_{j'}(t_{\ge k}) < k'$.

    By the choice of $k'$ and $t_{\ge k'}$, there is a job $d$ of current class $k'$ that is processed during $[t_{\ge k'},t_{\ge k'}+1]$. Let $t'$ denote the latest time before $t_{\ge k'}$ at which the algorithm makes $d$ ongoing. Note that $t'$ must exist because $k'$ exists in this case. So the job $d$ exists, and it must have become ongoing at some point in time.

    We next prove three separate claims before we continue the proof of this case.

    \begin{claim} \label{claim:fresh t' not processed}
         No job $b \in \afull(t')$ is processed during $[t',t^*]$.
    \end{claim}
    \begin{proof}
       At time $t'$, $d$ was just moved to the ongoing set, and thus all jobs in $\afull(t')$ have class $\ge k'$ at time $t'$. Since $d$ is processed at time $t_{\geq k'}$, none of these jobs in $\afull(t')$ has become ongoing by time $t_{\geq k'}$. So, $b$ cannot become ongoing during $[t',t_{\geq k'}]$. By definition of $t_{\geq k'}$, 
       following $t_{\geq k'}+1$
       the algorithm always processes an ongoing job of class less than $k'$. Since $b$'s class is at least $k'$, $b$ cannot become ongoing after $t_{\geq k'}$ either. The fact that $b$ does not become ongoing during $[t',\ts]$ implies that $b$ is also not processed during that interval.
    \end{proof}

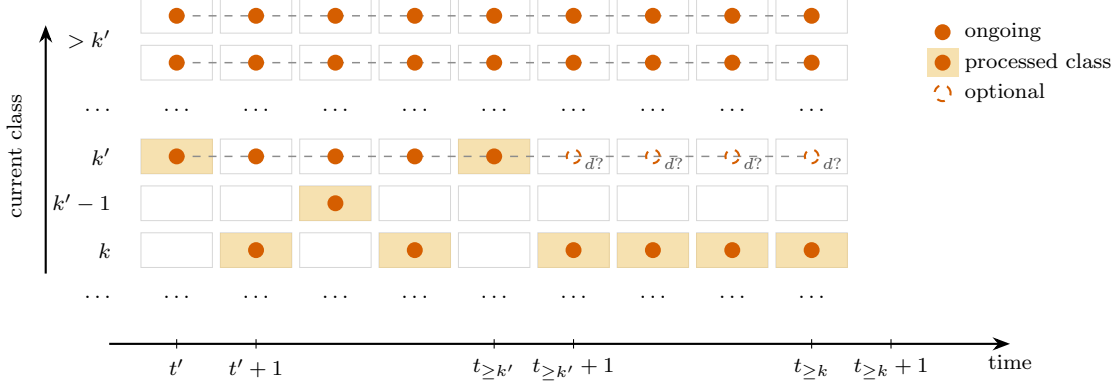
\begin{figure}
    \centering
    \begin{tikzpicture}[
        x=1.05cm,y=0.62cm,
    ]

\foreach \x in {0,...,8}{
        \foreach \k in {1,2,3,5,6}{
            \draw[algCell] (\x-.45,\k-.37) rectangle (\x+.45,\k+.37);
        }
    }

\foreach \x/\k in {0/3,1/1,2/2,3/1,4/3,5/1,6/1,7/1,8/1}{
        \fill[algProcCell] (\x-.45,\k-.37) rectangle (\x+.45,\k+.37);
    }

\foreach \x in {0,...,8}{
        \node at (\x,4) {\scriptsize $\hdots$};
        \node at (\x,0) {\scriptsize $\hdots$};
    }

\node[left=6pt] at (-.50,5.5) {\scriptsize $>k'$};
    \node[left=6pt] at (-.50,4) {\scriptsize $\hdots$};
    \node[left=6pt] at (-.50,3) {\scriptsize $k'$};
    \node[left=6pt] at (-.50,2) {\scriptsize $k'-1$};
    \node[left=6pt] at (-.50,1) {\scriptsize $k$};
    \node[left=6pt] at (-.50,0) {\scriptsize $\hdots$};
    \draw[-{Stealth}, thick] (-1.65,.5) -- (-1.65,5.8);
    \node[rotate=90] at (-2.05,3) {\scriptsize current class};

\foreach \y in {5,6}{
        \draw[algStayArrow] (0,\y) -- (4,\y);
        \draw[algStayArrow] (4,\y) -- (8,\y);
        \foreach \x in {0,...,8}{\algOngoingDot{\x}{\y}{0}{0}}
    }

\draw[algStayArrow] (0,3) -- (4,3);
    \draw[algStayArrow] (4,3) -- (8,3);
    \algOngoingDot{0}{3}{0}{0}
    \foreach \x in {1,2,3}{\algOngoingDot{\x}{3}{0}{0}}
    \algOngoingDot{4}{3}{0}{0}
    \algOptionalDot{5}{3}{0}{0}{d?}
    \algOptionalDot{6}{3}{0}{0}{d?}
    \algOptionalDot{7}{3}{0}{0}{d?}
    \algOptionalDot{8}{3}{0}{0}{d?}

\foreach \x/\y in {1/1,2/2,3/1,5/1,6/1,7/1}{
        \algOngoingDot{\x}{\y}{0}{0}
    }

\algOngoingDot{8}{1}{0}{0}

    \draw[-{Stealth}, thick] (-.85,-1) -- (10.5,-1) node[below,font=\scriptsize] {time};
    \foreach \x/\label in {
        0/{$t'$},
        1/{$t'+1$},
        4/{$t_{\ge k'}$},
        5/{$t_{\ge k'}+1$},
        8/{$t_{\ge k}$},
        9/{$t_{\ge k}+1$}
    }{
        \draw[] (\x,-.90) -- (\x,-1.10) node[below,font=\scriptsize] {\label};
    }

\begin{scope}[shift={(9.65,5.65)}]
        \fill[ongoingred] (0,0) circle[radius=3pt];
        \node[anchor=west] at (.15,0) {\scriptsize ongoing};
        \fill[algProcCell] (-.2,-.65-.3) rectangle (.2,-.65+.3);
        \fill[ongoingred] (0,-.65) circle[radius=3pt];
        \node[anchor=west] at (.15,-.65) {\scriptsize processed class};
        \draw[ongoingred,dashed,line width=.8pt] (0,-1.3) circle[radius=3pt];
        \node[anchor=west] at (.15,-1.3) {\scriptsize optional};
    \end{scope}

    \end{tikzpicture}
    \caption{Ongoing jobs in the situation of \Cref{eq:claim:4}; for clarity, we omit fresh jobs from the figure. The old ongoing jobs persist from $t'$ to $t_{\ge k}$; job $d$ is processed at $t'$ and during $[t_{\ge k'},t_{\ge k'}+1]$ and may be gone afterwards. The intermediate boxes illustrate processed ongoing jobs in two different classes below $d$. 
    Note that $d$ might complete at time $t_{\ge k'}+1$.
    During $[t_{\ge k},t_{\ge k}+1]$ at most one additional lower-class ongoing job can appear. The horizontal dashed lines indicate that the ongoing job stays the same.}
    \label{fig:inc-ongoing}
\end{figure}

    \begin{claim}
        It holds that $|\afull(t')| \leq |\afull(t_{\ge k})| - 2$. \label{eq:claim:2}
    \end{claim}
    \begin{proof}
    We have the following three observations:
    \begin{itemize}
        \item $\afull(t') \subseteq \afull(t_{\ge k})$. That is, every fresh job $b \in \afull(t')$ remains fresh at time $t_{\geq k}$, which follows from \Cref{claim:fresh t' not processed}.   
        
\item $j \in \afull(t_{\ge k})\setminus\afull(t')$. Since $ j \in \afull(t_{\ge k})$, it remains to argue that $j \not\in \afull(t')$. Suppose $j \in \afull(t')$. By \Cref{claim:fresh t' not processed}, $j$ is not processed after $t'$, and thus $\fint_j \le t'$, contradicting that $\fint_j \in [t_{\ge k'}+1,t_{\ge k}]$.

        \item $i \in \afull(t_{\ge k}) \setminus \afull(t')$. It holds that $i \in \afull(t_{\ge k})$ because the algorithm processed job $j'$ of current class $k_{j'}(t_{\ge k}) \ge k$ during $[t_{\ge k},t_{\ge k}+1]$. Furthermore, we have $i \not\in \afull(t')$, 
        because if $i \in \afull(t')$, then $i$ would have become ongoing before $d$ at that time because $k_i(t') \le k_i(t_{\ge k}) < k < k'$.
    \end{itemize}
    In summary, these three observations give the claimed inequality.
    \end{proof}

    \begin{claim}
    \label{eq:claim:4}
        It holds that $|\apart(t_{\ge k})| \le |\apart(t')| + 1$.
    \end{claim}

    A visualization of the situation for this claim is given in \Cref{fig:inc-ongoing}.

   \begin{proof}
        First note that job $d$ has the smallest current class among ongoing jobs at time $t'$, because then it became ongoing, and at time $t_{\ge k'}$, because then it has been processed.
        By our choice of the time $t'$, job $d$ is ongoing during $[t',t_{\ge k'}+1]$. Thus, no job in $\aon(t') \setminus\{d\}$ is processed during this time interval and no ongoing job of larger current class than $d$ can become ongoing during this interval. Therefore, we can conclude that $\apart(t') = \apart(t_{\ge k'})$.

         We next focus on the time between $t_{\ge k'}+1$ and $t_{\ge k}$. Since we process job $j'$ at time $t_{\ge k}$ with class $k_{j'}(t_{\ge k}) \ge k$, it cannot be that we process a job $j''$ during $[t_{\ge k'}+1, t_{\ge k}]$ of current class $> k_{j'}(t_{\ge k})$, as this would contradict our choice of $k'$.
         Thus, since $j'$ has the smallest class of ongoing jobs at time $t_{\ge k}$ (and thus all other ongoing jobs have a strictly larger class), all other ongoing jobs at time $t_{\ge k}$ must have been already ongoing at time $t_{\ge k'}$. Thus, in total we have
         \[
         |\apart(t_{\ge k})| \le |\apart(t_{\ge k'})| + 1 = |\apart(t')| + 1
         \]
         as claimed. 
   \end{proof}

Finally, we can show that \eqref{eq:claim:3} holds:
    \begin{align*}
        \frac{1}{4} |A(t_{\ge k})| - |\afull(t_{\ge k})| 
        &= \frac14 |\aon(t_{\ge k})| - \frac34 |\afull(t_{\ge k})| \\
        &\le \frac14 |\aon(t')| + \frac14 - \frac34 |\afull(t_{\ge k})| \\
        &\le \frac14 |\aon(t')| + \frac14 - \frac34 (|\afull(t')|+2) \\
        &= \frac14 |A(t')| - |\afull(t')| + \frac14 - \frac32 
        \le 1 + \frac14 - \frac32 \leq 0 \ ,
    \end{align*}
    where the first inequality is by \Cref{eq:claim:4}, the second inequality by \Cref{eq:claim:2},
    and the final inequality uses $\frac{1}{4} |A(t')| - |\afull(t')|  \le 1$, which follows from \Cref{lem:many-fresh}.
    This concludes \Cref{eq:claim:3} and thus completes the proof of the claim.
\end{proof}

With~\Cref{claim:strong:excess} in place, we are ready to prove the strong excess lemma.

\begin{proof}[Proof of~\Cref{lem:strong-excess}]
    Our goal is to show the following two properties, which imply the lemma:
    \begin{enumerate}
        \item All jobs that are processed during $[t_{\ge k}+1,\ts]$ are released during $[t_{\ge k}+1,\ts]$. Therefore, all these jobs are members of $S(t_{\geq k}+1)$.
        \item All jobs $j \in \cup_{k'<k}\bar{F}(k')$ were released during $[t_{\geq k}+1,\ts]$. Hence, $\cup_{k'<k}\bar{F}(k') \subseteq S(t_{\geq k}+1)$. 
    \end{enumerate}

    Before we prove these two properties, we show that they imply the lemma.  

    By the first property above, the algorithm only works on jobs in $S(t_{\geq k}+1)$ during $[t_{\ge k}+1,\ts]$. Since we assume that the algorithm does not idle during $[t_{\ge k}+1,\ts]$, this implies 
    \begin{align} \label{eq:total elapsed eq}
        \sum_{j \in S(t_{\geq k}+1)} e_{j}(\ts) =  (\ts - (t_{\ge k}+1)).
    \end{align}
    Plugging this equation into the definition of effective excess yields:
    \begin{align}
    \redex(S(t_{\ge k}+1))
    &\geq \Bigl( \sum_{j \in S(t_{\geq k}+1)} \redp_j \Bigr) - (\ts - (t_{\ge k}+1)) \nonumber\\
    &= \sum_{j \in S(t_{\geq k}+1)} (e_{j}(\ts) + \redp_{j}(\ts)) - (\ts - (t_{\ge k}+1)) \notag \\
    &\overset{(\ref{eq:total elapsed eq})}{=} \sum_{{j} \in S(t_{\geq k}+1)} \redp_{j}(\ts) \notag \\
    &\geq \sum_{k' < k} \sum_{j \in \bar{F}(k')} \redp_{j}(\ts) \label{eq:strong:excess:1}
    \end{align}
    where
\eqref{eq:strong:excess:1} follows from the second property above.
    It remains to show the two properties. Note that $F(k) \neq \emptyset$ since $\bar F(k) \neq \emptyset$ and $\bar F(k) \subseteq F(k)$ (\Cref{obs:bar F subseteq}), and thus \Cref{claim:strong:excess} can be applied.  

    To prove the first property, note that for all jobs $j \in A(t_{\ge k})$, we have $k_j(t') \ge k_j(t_{\ge k}) \geq k$ for all $t' \in [t_{\ge k}, \ts]$ by \Cref{claim:strong:excess} and the fact that the current class is monotone.
    Thus, no job in $A(t_{\ge k})$ is processed during $[t_{\ge k}+1,\ts]$ by the definition of $t_{\ge k}$. Hence, all jobs processed during that time interval must be released within it.

    For the second property, consider an arbitrary $j \in \bar{F}(k')$ where $k' < k$.
    Thus $k' = \redk_j$ and $j \in \bar F(\redk_j)$.
    For the sake of contradiction, assume that $r_j \le t_{\ge k}$. 

    \begin{claim}
        $\fint_j = t_{\ge k}+1$.
    \end{claim}
    \begin{proof}
        By~\Cref{claim:strong:excess}, we have $k_j(t_{\ge k}) \geq k$ and thus $\fink_j = k_j(\ts-1) \geq k$.  We prove the following two inequalities. 
        \begin{itemize}
            \item $(\fint_j \geq t_{\ge k}+1)$. Otherwise, the effective class of $j$ is \[k >  k' = \redk_{j} \ge  \min \Bigl\{\fink_j,\max_{\fint_j \le t < \ts} k(t) \Bigr\} \ge \min \Bigl\{\fink_j,\max_{t_{\ge k} \le t < \ts} k(t) \Bigr\} \ge k,\] which is a contradiction. The last inequality follows since $k_j(t_{\ge k}) \geq k$ and $\fink_j \geq k$.
            \item $(\fint_j \leq t_{\ge k}+1)$. First note that since $r_j \le t_{\ge k} < t_{\ge k} + 1 \le \fint_j$, job $j$ cannot reach its final class upon release.
    Thus, job $j$ is processed during timestep $[\fint_j-1,\fint_j]$. 
    Now suppose that $\fint_j > t_{\ge k}+1$, that is $\fint_j-1 > t_{\ge k}$. Since $k_j(t_{\ge k}+1) \ge k$, this means that $j$ is processed after time $t_{\ge k}+1$ while having current class $\ge k$, a contradiction to the definition of $t_{\ge k}$.\qedhere \end{itemize}
    \end{proof}

In summary, we have the following three properties:
    \begin{enumerate}
        \item [(i)] $\bar{F}(k) \not= \emptyset$,
        \item [(ii)]$j$ becomes fresh at $\fint_j= t_{\ge k}+1$; thus, $j$ is processed at time $[t_{\geq k}, t_{\geq k}+1]$ and not touched during $[t_{\ge k}+1,\ts]$, and
        \item  [(iii)] $j \in \bar{F}(\redk_j)$.
    \end{enumerate}
    However, by construction of the sets $\bar{F}$, these three properties imply $j \not\in \bigcup_{k''} \bar{F}(k'')$; a contradiction. We thus conclude that the second property holds. Note that this is exactly the reason for pruning the sets $F(k)$ into the set $\bar{F}(k)$ that was outlined in~\Cref{sec:consistent fresh jobs}.
\end{proof}

\subsubsection{The Weak Excess Lemma}

\begin{lemma}[Weak Excess Lemma]
\label{lem:weak-excess}
If $F(k) \not= \emptyset$ then 
$\redex(S(t_{>k}+1)) \ge
        \sum_{k'\le k}\sum_{j \in \bar F(k')} \redp_{j}(\ts)
        - 3 \cdot \mu_2 \cdot \Lambda_k$.
\end{lemma}

\begin{proof}
We first establish that there can exist at most one job $j$
at time $t_{>k}$ of current class $k_{j}(t_{> k}) \leq k$.
By definition of $t_{> k}$, the algorithm processes a job $j'$ of current class $k_{j'}(t_{> k}) > k$ during $[t_{>k}, t_{>k}+1]$.
Thus, if $j$ exists, then it is fresh at time $t_{> k}$ by \Cref{lem:stack-monotone}.
By \Cref{lem:non-greedy}, the existence of $j$ and $j'$ at time $t_{> k}$ implies that $j$ is the only job of current class $\leq k$ at time $t_{> k}$.

We show the following two properties:
\begin{enumerate}
    \item All jobs $j'' \not= j$ that are processed during $[t_{>k}+1,\ts]$ are released during $[t_{>k}+1,\ts]$. Thus, the corresponding jobs are part of $S(t_{>k}+1)$.
    \begin{itemize}
        \item {\bf Proof.}  Note that current classes (actual classes $k_j(\cdot)$, not effective classes) of jobs only increase.  Since $j$ is the only job active at time $t_{>k}$ which can have class $\le k$, the definition of $t_{>k}$ implies that $j$ is the only job active at $t_{>k}$ that can be processed during $[t_{>k}+1,\ts]$.  Hence, all other jobs that are processed during $[t_{>k}+1,\ts]$ are released during $[t_{>k}+1,\ts]$.
    \end{itemize}
    \item All jobs  $j'' \in \left(\bigcup_{k' \le k} \bar{F}(k')\right) \setminus \{j\}$ belong to the set $S(t_{>k}+1)$ with possibly a single exception job $d$. In particular, $j$ and $d$ are the only two jobs that can be members of $\left(\bigcup_{k' \le k} \bar F(k')\right)\setminus S(t_{>k}+1)$.

    \begin{itemize}
        \item {\bf Proof.}  Consider an arbitrary job $j'' \in \bigcup_{k' \le k} \bar{F}(k')$ with $j'' \not= j$ and assume that $j'' \not\in S(t_{>k}+1)$. That is, $r_{j''} \le t_{>k}$. As argued above, $j$ is the only job of current class $k_{j}(t_{> k}) \le k$ at time $t_{> k}$. Hence, $\fink_{j''} \ge k_{j''}(t_{> k}) > k$. If $j''$ reaches its final class strictly before $t_{> k} +1$, that is, $t_{j''} \le t_{> k}$, then $\redk_{j''} > k$ as $\fink_{j''} > k$ and the algorithm processes a job of class $> k$ during $[t_{> k}, t_{> k} +1]$; this is a contradiction.
        Hence, if $j'' \in \bar{F}(k')$ with $k' \le k$ was released before $t_{> k} +1$, then it has to reach its final class fresh during $[t_{> k} +1,\ts]$. However, this can only happen if $j''$ is the job processed during $[t_{> k}, t_{> k} +1]$. That is, we must have $j''=j'$.
        This implies that at most one job $d \in \left(\bigcup_{k' \le k} \bar{F}(k')\right) \setminus \{j\}$ is released outside of the interval $[t_{> k} +1,\ts]$.
        In particular, this means that $j$ and $d$ are the only two jobs which can be part of $\left(\bigcup_{k' \le k} \bar F(k')\right)\setminus S(t_{>k}+1)$.
    \end{itemize}
\end{enumerate}

We next show that the two properties indeed imply the lemma.
By the definition of excess, $\redex(S(t_{>k}+1)) = \max\{0, \redp(S(t_{>k}+1)) - (\ts - (t_{> k } + 1))\}$.
Let $\Delta_j$ denote the amount of time during $[t_{> k} + 1, \ts]$ when the algorithm works on $j$; we define $\Delta_j = 0$ if $j$ does not exist.
By the first property above and using that the algorithm does not idle, we have
\[
 \Delta_j + \sum_{i \in S(t_{>k}+1)} e_{i}(\ts) = \ts - (t_{>k}+1).
\]
This gives us
\begin{align}
\Delta_j + \redex(S(t_{>k}+1))
&\geq \Delta_j + \redp(S(t_{>k}+1)) - (\ts - (t_{>k}+1)) \nonumber\\
&= \Delta_j + \sum_{i \in S(t_{>k}+1)} (e_{i}(\ts) + \redp_{i}(\ts)) - (\ts - (t_{>k}+1)) \nonumber\\
&= \sum_{i \in S(t_{>k}+1)} \redp_{i}(\ts) \nonumber\\
&\geq \Bigl(\sum_{k' \le k} \sum_{i \in \bar F(k')} \redp_{i}(\ts)\Bigr) - \redp_{d}(\ts) - \redp_{j}(\ts)\label{eq:weak:excess:1}\\
&\ge \Bigl(\sum_{k' \le k} \sum_{i \in \bar F(k')} \redp_{i}(\ts)\Bigr) - 2 \eps \cdot \Lambda_k \label{eq:weak:excess:2}\\
&\ge \Bigl(\sum_{k' \le k} \sum_{i \in \bar F(k')} \redp_{i}(\ts)\Bigr) - 2 \mu_2 \cdot \Lambda_k.\label{eq:weak:excess:3}
\end{align}
Here inequality~\eqref{eq:weak:excess:1} follows from the second property above, the inequality~\eqref{eq:weak:excess:2} uses \Cref{lem:p-class-bounds} and $\redk_d, \redk_j\le k$, and the inequality~\eqref{eq:weak:excess:3} uses $\eps \le 1\le \mu_2$. We remark that there is the possibility that one or both jobs in $\{j,d\}$ do not exist or are not part of $( \bigcup_{k' \le k} \bar F(k') ) \setminus S(t_{>k}+1)$. In that case the inequality above only becomes stronger as we do not have to subtract $\redp_d(\ts)$ or $\redp_j(\ts)$.

We continue by bounding $\Delta_j$. If $j$ is not processed during $[t_{> k} + 1, \ts]$, then we are done. Hence, assume that $j$ is processed during $[t_{> k} + 1, \ts]$, and let $k'$ denote the maximum current class that $j$ reaches strictly before time $\ts$.
Since $j$ is processed during $[t_{> k} + 1, \ts]$, the definition of $t_{> k}$ implies $k' \leq k$.
By the definition of the current classes, this implies $\Delta_j \leq \mu_2 \Lambda_k$: if the class of $j$ is governed by elapsed processing, then the elapsed processing of $j$ is at most $\Lambda_k$; otherwise, the estimate gives $p_j\le \mu_2\hat p_j\le \mu_2\Lambda_k$.

We conclude
\[
\redex(S(t_{>k}+1)) \ge \Bigl(\sum_{k' \le k} \sum_{i \in \bar F(k')} \redp_{i}(\ts)\Bigr) - 3 \cdot  \mu_2 \cdot  \Lambda_k.\qedhere
\]
\end{proof}

\section{Dual Fitting} \label{sec:dual fitting}

We now present the full dual fitting argument.
Let $\rho := \frac{\mu_1 \cdot \mu_2}{\eps}$.
We call a class $k$ \emph{crucial} 
if $\bar F(k) \neq \emptyset$.
Let $I$ denote the set of crucial classes.
Furthermore, let $I_{l} := \{k \in I : |\bar F(k)| \leq 7 \cdot \rho \}$
and $I_{h} := \{k \in I : |\bar F(k)| > 7 \cdot \rho \}$
denote the \emph{light} and \emph{heavy} crucial classes, respectively.
We define for each crucial class $k \in I$
the reference set
\[
    S_k := \begin{cases}
        S(t_{\ge k}+1) &\text{if } k \in I_l, \text{ and} \\
        S(t_{>k}+1) &\text{if } k \in I_h \ .
    \end{cases}
\]
Our next goal is to establish that $\redex(S_k) > 0$ for almost all crucial classes, which we need to define a dual solution.
To this end, we first show the following preliminary lemma,
which we will use to further lower bound the volume certificates given by \Cref{lem:strong-excess} and \Cref{lem:weak-excess}.

\begin{lemma}\label{lem:full-lower-bounds}
    If $j \in F(k)$, then $\redp_{j}(\ts) \geq\frac{\eps}{\mu_1} \cdot \lambda_k$.
\end{lemma}

\begin{proof}
    From $j \in F(k)$, we get $j \in \afull(\ts)$ and $\redk_j=k$.
    By~\Cref{def:effective:time}, we have $\redp_j = e_j(\fint_j) + \eps \cdot \min\{p_j, \Lambda_{\redk_j}\}$ and by~\Cref{def:effective:class}, we have $\redk_j = \min\{\fink_j,\max_{\fint_j \le t < \ts} k(t)\}$. Since $j \in \afull(\ts)$ also implies that $j$ is not touched after reaching its final class, we get
    \[
    \redp_j(\ts) = \redp_j - e_j(\ts) = \redp_j - e_j(\fint_j) = \eps \cdot \min\{p_j, \Lambda_{\redk_j}\}.
    \]
    Thus it is enough to prove that $p_j \ge \frac{1}{\mu_1}\lambda_{\redk_j}$, because $\Lambda_{\redk_j}\ge \lambda_{\redk_j}$.

    Suppose first that $j$ is still unknown at time $\fint_j$. If $\redk_j=0$, then $p_j\ge 1=\lambda_{\redk_j}$. Otherwise, $j$ reaches its final class $\fink_j$ because the elapsed time of $j$ reaches at least $\lambda_{\fink_j}$. Hence, $p_j \ge \lambda_{\fink_j} \ge \lambda_{\redk_j}$.

    Now suppose that $j$ is known at $\fint_j$. Then, by~\Cref{lem:final:class:pred}, it just became known at time $\fint_j$. If $e_j(\fint_j) \ge \lambda_{\fink_j}$, then we can argue in the same way as in the unknown case. Hence, assume $e_j(\fint_j) < \lambda_{\fink_j}$. This implies that the arrival of the estimate promoted the class of the job, and thus $\fink_j = 1+\floor{\log_2 \hat{p}_j}$.
    By definition of the problem, it holds that $p_j \ge \frac{1}{\mu_1} \cdot \hat{p}_j \ge  \frac{1}{\mu_1} \cdot \lambda_{\fink_j} \ge \frac{1}{\mu_1} \cdot \lambda_{\redk_j}$.

    We conclude
    \[
    \redp_j(\ts)=\eps \cdot \min\{p_j,\Lambda_{\redk_j}\} \ge
    \frac{\eps}{\mu_1} \lambda_{\redk_j}=\frac{\eps}{\mu_1} \lambda_k
    . \qedhere
    \]
\end{proof}

We next give a nonzero lower bound on the effective excess of heavy crucial classes.

\begin{lemma}[Excess of heavy classes]
\label{lem:large-class-excess}
    If $k \in I_h$, then $\redex(S_k) \geq \frac{|\bar F(k)|}{7} \frac{\eps}{\mu_1} \lambda_k > 0$.
\end{lemma}

\begin{proof}
The assumption of the lemma implies $S_k = S(t_{>k}+1)$.
We have
\begin{align}
        \redex(S_k) =  \redex(S(t_{>k}+1))
        &\geq - 3\mu_2\Lambda_k + \sum_{k' \leq k} \sum_{j \in \bar F(k')} \redp_{j}(t^*) \label{eq:heavy-ex:1} \\
        &\geq - 3\mu_2\Lambda_k + \sum_{j \in \bar F(k)} \redp_{j}(t^*) \notag \\
        &\geq - 6\rho \lambda_k \frac{\eps}{\mu_1} + | \bar F(k)| \frac{\eps}{\mu_1} \lambda_k  = \left( |\bar F(k)| - 6\rho \right) \frac{\eps}{\mu_1} \lambda_k \label{eq:heavy-ex:3}\\
        &\geq \frac{|\bar F(k)|}{7} \frac{\eps}{\mu_1} \lambda_k \, , \label{eq:heavy-ex:5}
    \end{align}
    where 
    \eqref{eq:heavy-ex:1} uses \Cref{lem:weak-excess}, 
    \eqref{eq:heavy-ex:3} uses \Cref{lem:full-lower-bounds}, $\Lambda_k \le 2\lambda_k$, and $\rho=\mu_1\mu_2/\eps$,
    and \eqref{eq:heavy-ex:5} holds because $|\bar F(k)| > 7\rho$ as $k \in I_h$.
\end{proof}

Finally, we show that almost all light classes have positive excess.

\begin{lemma}[Excess of light classes is positive]
\label{lem:small-class-excess}
    For every $k \in I_l \setminus \{\min (I_l)\}$, we have
    $\redex(S_k) > 0$.
\end{lemma}

\begin{proof}
  The assumption of the lemma implies $S_k = S(t_{\ge k}+1)$.  Using
  \Cref{lem:strong-excess} and \Cref{lem:full-lower-bounds}, we can
  conclude that
    \begin{align*}
        \redex(S_k) =  \redex(S(t_{\ge k}+1)) 
        \geq \sum_{k' < k} \sum_{j \in \bar F(k')} \redp_{j}(t^*) 
        \geq \sum_{j \in \bar F(\min (I_l))} \redp_{j}(t^*) 
        \geq \sum_{j \in \bar F(\min (I_l))} \frac{\eps}{\mu_1}
      \lambda_{\min(I_l)} > 0 \ .
    \end{align*}
    Hence the proof.
\end{proof}

\subsection{Definition of Dual Solution}
\label{sec:defin-dual}

We now define a dual solution to $(\dppart)$.
If $I_l\neq \emptyset$, let $k_1:=\min(I_l)$ and $I' := I\setminus\{k_1\}$. If $I_l=\emptyset$, let $I':=I$. This technical detail is necessary because \Cref{lem:small-class-excess} does not guarantee positive excess if the smallest crucial class is light; this is a property that we need for defining a dual solution.
For each subset of jobs $S$, we set
\[
y_S := \sum_{\substack{k \in I' \\S_k = S}} \frac{|\bar F(k)|}{7 \rho \cdot \redex(S)} \ . 
\]
\Cref{lem:large-class-excess,lem:small-class-excess} show that $y_S$ is well-defined, as $\rho > 0$ and $\redex(S_k) > 0$ for all $k \in I'$.

\subsection{Objective Value}
\label{sec:objective-value}

We first show that our dual assignment captures a constant fraction of the number of active jobs in the schedule of the algorithm at time $\ts$.

\begin{lemma}[Dual Objective Value]\label{lem:dual-obj}
    We have that $\sum_{S} \redex(S) \cdot y_S \geq \frac{1}{56\rho} |A(\ts)| - \frac{1+14\rho}{14\rho}$.
\end{lemma}

\begin{proof}

We compute
    \begin{align}
        \sum_S \redex(S) \cdot y_S
        &=
        \sum_S \redex(S)\sum_{\substack{k \in I'\\S_k = S}}  \frac{|\bar F(k)|}{7 \rho \cdot \redex(S)}  \label{eq:dualobj1} \\
        &=
        \sum_{k \in I'} \redex(S_k) \cdot \frac{|\bar F(k)|}{7 \rho \cdot \redex(S_k)} \notag  
        =
        \sum_{k \in I'} \frac{|\bar F(k)|}{7 \rho} \notag \\
        &\geq
        \Bigl(\sum_{k \in I} \frac{|\bar F(k)|}{7 \rho} \Bigr) - 1 \label{eq:dualobj2} \\
        &\geq
        \frac{1}{14\rho} |\afull(\ts)| - 1 \label{eq:dualobj3} \\
        &=
        \frac{1}{14\rho} \left( \frac14 |A(\ts)| - 1 \right) - 1
        =
        \frac{1}{56\rho} |A(\ts)| - \frac{1+14\rho}{14\rho} \, , \label{eq:dualobj4}
    \end{align}
where \eqref{eq:dualobj1} uses the definition of our dual solution, \eqref{eq:dualobj2} uses the definition of $I'$ and the fact that we remove at most one light class, which has size at most $7\rho$,
\eqref{eq:dualobj3} follows from \Cref{lem:many-in-bar-F}, and \eqref{eq:dualobj4} follows from \Cref{lem:many-fresh}.
\end{proof}

\subsection{Approximate Feasibility}
\label{sec:appr-feas}

Our goal is to show that our defined dual solution 
is approximately feasible for $(\dppart)$.
First note that $y_S \geq 0$ for all $S \subseteq J$.
Thus, it remains to show that the dual constraint of $(\dppart)$ is satisfied subject to appropriate scaling.

To this end, fix a job $j^*$ with $r_{j^*} \le \ts$. 
We refine $I_l$ and $I_h$ to classes which are relevant for $j^*$'s dual constraint:
\begin{itemize}
\item Let $k^*_l$ denote the smallest class $k \in I_l \cap I'$ such that $j^* \in S_k$, if such a class exists. If $k^*_l$ exists, define
$I_l(j^*) = \{k \in I_l \cap I' : k \geq k^*_l \text{ and } j^* \in S_{k}\}$; otherwise, define $I_l(j^*)=\emptyset$.
\item Let $k^*_h$ denote the smallest class $k \in I_h$ such that $j^* \in S_k$, if such a class exists. If $k^*_h$ exists, define
$I_h(j^*) = \{k \in I_h : k \geq k^*_h \text{ and } j^* \in S_{k}\}$; otherwise, define $I_h(j^*)=\emptyset$.
\end{itemize}
This setup helps us to prove the following preliminary lemma. It gives an upper bound on $\redp_{j^*}$, which appears in the dual constraint.

\begin{lemma}
    \label{lem:size:lb}
    If $k^*_l$ exists, then
    $\redp_{j^*} \le  \mu_2 \cdot \Lambda_{k^*_l}$.
    If $k^*_h$ exists, then
    $\redp_{j^*} \le  2\mu_2 \cdot \Lambda_{k^*_h}$.
\end{lemma}

\begin{proof}
We show that job $j^*$ satisfies the following two properties:
\begin{enumerate}
    \item If $k\ge 1$ and $j^* \in S(t_{\ge k}+1)$, then $\redp_{j^*} \le  \mu_2 \cdot \Lambda_k$.
    \item If $j^* \in S(t_{> k}+1)$, then $\redp_{j^*} \le 2\mu_2 \cdot \Lambda_k$.
\end{enumerate}
The lemma then follows by definition of $k^*_l$ and $k^*_h$; note that any light class in $I'$ is at least $1$, since the smallest light class is not included in $I'$.
The second property immediately follows from the first property as $t_{>k} = t_{\geq k + 1}$.

We now show the first property. If $j^* \in S(t_{\ge k}+1)$, then $r_{j^*} \in [t_{\ge k}+1,\ts]$. By definition of $t_{\ge k}$, no job of current class at least $k$ is processed after time $t_{\ge k}$.
    \begin{itemize}
        \item If $j^* \in \apart(\ts)$, then $\redp_{j^*} = e_{j^*}(\ts)$ and $\redk_{j^*} = \fink_{j^*}$ by~\Cref{def:effective:class} and~\Cref{def:effective:time}. Furthermore, $j^* \in \apart(\ts)$ implies that $j^*$ is processed after reaching its final class $\fink_{j^*}$. Hence $\fink_{j^*} < k$. If $j^*$ is still unknown at $\ts$, this implies $e_{j^*}(\ts) \le \Lambda_{k-1} \le \Lambda_k$. If $j^*$ is known at $\ts$, then $p_{j^*} \le \mu_2 \hat p_{j^*} \le \mu_2 \Lambda_{k-1} \le \mu_2 \Lambda_k$. In either case, $\redp_{j^*}\le \mu_2\Lambda_k$.

        \item If $j^* \not\in A(\ts)$, then $\redp_{j^*} = p_{j^*}$ and $\redk_{j^*} = \fink_{j^*}$. Since $j^*$ is run to completion during $[t_{\ge k}+1,\ts]$, it cannot be processed in current class at least $k$ after time $t_{\ge k}$, and hence $\fink_{j^*}<k$. The same class bound as in the previous case gives $p_{j^*}\le \mu_2\Lambda_k$.

        \item If $j^* \in \afull(\ts)$, then $\redk_{j^*} = \min\{\fink_{j^*}, \max_{\fint_{j^*}\le t < \ts } k(t)\}$ and $\redp_{j^*} = e_{j^*}(\ts)+ \eps \cdot \min\{p_{j^*}, \Lambda_{\redk_{j^*}}\}$.
        Since $\fint_{j^*} \ge r_{j^*} \ge t_{\ge k} + 1$ and by definition of $t_{\ge k}$, we have $\redk_{j^*} < k$ and $e_{j^*}(\ts)\le \Lambda_{k-1}$. Thus,
        \[
        \redp_{j^*} \le \Lambda_{k-1}+\eps \Lambda_{k-1} \le \Lambda_k \le \mu_2 \Lambda_k .
        \]
    \end{itemize}
This completes the proof of the lemma.
\end{proof}

We now come to the actual proof of the approximate dual feasibility. We split the sum in the dual constraint into the contribution of
sets in $S_k$ with $k \in I_h(j^*)$ (\Cref{lem:dual-feasibility1})
and with $k \in I_l(j^*)$ (\Cref{lem:dual-feasibility2}), and bound them separately in each of these lemmas.

\begin{lemma}\label{lem:dual-feasibility1}
It holds that
   $\sum_{k \in I_h(j^*)} \min\{\redex(S_{k}), \redp_{j^*}\} \cdot \frac{|\bar F(k)|}{7\rho \cdot \redex(S_{k})} \leq 8$.
\end{lemma}

\begin{proof}
    If $I_h(j^*)=\emptyset$, the claim is immediate. Hence assume that $k^*_h$ exists.
    For each $k \in I_h(j^*) \subseteq I_h$, \Cref{lem:large-class-excess} gives
    \[
    \redex(S_{k}) \geq \frac{|\bar F(k)|}{7} \frac{\eps}{\mu_1} \lambda_k \, .
    \]
    Using this and \Cref{lem:size:lb}, we can conclude that
    \begin{align*}
        \sum_{k \in I_h(j^*)} \min\{\redex(S_{k}), \redp_{j^*}\} \frac{|\bar F(k)|}{7\rho \cdot \redex(S_{k})}
        &\leq \sum_{k \geq k^*_h} \redp_{j^*} \frac{|\bar F(k)|}{7\rho \frac{|\bar F(k)|}{7} \frac{\eps}{\mu_1} \lambda_k} \\
        &= \sum_{k \geq k^*_h} \redp_{j^*} \frac{1}{\mu_2 \lambda_k} 
        \leq \sum_{k \geq k^*_h} \frac{2\Lambda_{k^*_h}}{\lambda_k} \leq 8 \, .
    \end{align*}
    This completes the proof of the lemma.
\end{proof}

\begin{lemma}\label{lem:dual-feasibility2}
   It holds that
   $\sum_{k \in I_l(j^*)} \min\{\redex(S_{k}), \redp_{j^*}\} \cdot \frac{|\bar F(k)|}{7\rho \cdot \redex(S_{k})} \le 7$.
\end{lemma}

\begin{proof}
   If $I_l(j^*)=\emptyset$, the claim is immediate. Hence assume that $k^*_l$ exists.
   For each $k \in I_l(j^*)$ we have that $S_k = S(t_{\ge k}+1)$.
   Define $T_\ell = \{k \in I_l(j^*) : |\bar F(k)| \in [2^\ell, 2^{\ell+1} - 1]\}$ for each $\ell \geq 0$
   and $\Gamma = \lfloor \log_2(7\rho) \rfloor$.
   We have
   \begin{align}
       \sum_{k \in I_l(j^*)} \min\{\redex(S_{k}), \redp_{j^*}\} \cdot \frac{|\bar F(k)|}{7\rho \cdot \redex(S_{k})} \notag
       &\leq
       \sum_{\ell=0}^\Gamma \sum_{k \in T_\ell}\min\{\redex(S_{k}), \redp_{j^*}\} \cdot \frac{2^{\ell+1}}{7\rho \cdot \redex(S_{k})} \notag \\
       &\leq
       \sum_{\ell=0}^\Gamma \frac{2^{\ell+1}}{7\rho} \sum_{k \in T_\ell} \min \left\{1, \frac{\redp_{j^*}}{\redex(S_{k})} \right\} \notag \\
       &\leq
       \sum_{\ell=0}^\Gamma \frac{2^{\ell+1}}{7\rho} \sum_{k \in T_\ell} \min \left\{1, \frac{\mu_2 \Lambda_{k^*_l}}{\sum_{k' < k} \sum_{j \in \bar F(k')} \redp_{j}(t^*)} \right\} \label{eq:feas1}
   \end{align}
   where the final inequality uses \Cref{lem:strong-excess}, because each $T_\ell \subseteq I$, and \Cref{lem:size:lb}.

   Now focus on some index $\ell \ge 0$ with $T_\ell\neq \emptyset$.
    Let $k_1,\ldots,k_{|T_\ell|}$ be the indices in $T_\ell$ arranged in increasing order.
    Then
    \[
        \sum_{i=1}^{|T_\ell|} \min \left\{ 1, \frac{\mu_2 \Lambda_{k^*_l}}{\sum_{k' < k_i} \sum_{j \in \bar F(k')} \redp_{j}(t^*)} \right\}
        \leq
        1 + \sum_{i=2}^{|T_\ell|} \min \left\{ 1, \frac{\mu_2 \Lambda_{k^*_l}}{\sum_{j \in \bar F(k_{i-1})} \redp_{j}(t^*)} \right\}
    \]
    Moreover, for each $k \in T_\ell$, we have
    \[
    \sum_{j \in \bar F(k)} \redp_{j}(t^*)
    \geq \frac{\eps}{\mu_1} |\bar F(k)| \cdot \lambda_k
    \geq \frac{\eps}{\mu_1} 2^\ell \lambda_k
    \]
    using \Cref{lem:full-lower-bounds}.
    Thus,
    \begin{align*}
        \sum_{k \in T_\ell} \min \left\{1, \frac{\mu_2 \Lambda_{k^*_l}}{\sum_{k' < k} \sum_{j \in \bar F(k')} \redp_{j}(t^*)} \right\}
        &\leq
        1 + \sum_{i=2}^{|T_\ell|} \frac{\mu_2 \Lambda_{k^*_l}}{\sum_{j \in \bar F(k_{i-1})} \redp_{j}(t^*)} \\
        &\leq
        1 + \sum_{i=2}^{|T_\ell|}  \frac{\mu_2 \Lambda_{k^*_l}}{\frac{\eps}{\mu_1} 2^\ell \lambda_{k_{i-1}} } \\
        &=
        1 + \rho \frac{\Lambda_{k^*_l}}{2^\ell} \sum_{i=2}^{|T_\ell|}  \frac{1}{\lambda_{k_{i-1}}} \\
        &\leq
        1 + \rho \frac{2\Lambda_{k^*_l}}{2^\ell} \frac{1}{\lambda_{k_{1}}} .
    \end{align*}
    For the final step, we plug the above bound back into \eqref{eq:feas1}. Let $\mathcal L:=\{\ell\in\{0,\ldots,\Gamma\}:T_\ell\neq\emptyset\}$, and for each $\ell\in\mathcal L$ let $K_\ell := \min T_\ell$, that is, the smallest class $\geq k^*_l$ in $T_\ell$. Then we can conclude that
    \begin{align*}
        \sum_{\ell\in\mathcal L} \frac{2^{\ell+1}}{7\rho} \left( 1 + \rho \frac{2\Lambda_{k^*_l}}{2^\ell} \frac{1} {\lambda_{K_\ell}} \right)
        &\leq
        \sum_{\ell=0}^\Gamma \frac{2^{\ell+1}}{7\rho} + \frac{4\Lambda_{k^*_l}}{7} \sum_{\ell\in\mathcal L} \frac{1} {\lambda_{K_\ell}} \\
        &\leq \frac{2^{\Gamma+2}}{7\rho} +  \frac{4\Lambda_{k^*_l}}{7} \frac{2}{\lambda_{k^*_l}}
        \leq \frac{4 \cdot 7\rho}{7\rho} +  \frac{16}{7}
        = 4 + \frac{16}{7} < 7 .
    \end{align*}
    This completes the proof of the lemma.
\end{proof}

\begin{lemma}[Approximate Dual Feasibility]\label{lem:dual-feasibility}
    It holds that $\sum_{S: j^* \in S} \min( \redp_{j^*}, \redex(S) ) \cdot y_S \leq 15$.
\end{lemma}

\begin{proof}
By the definitions of $I_l(j^*)$ and $I_h(j^*)$, for all $k \in I' $ with $j^* \in S_k$ we have that $k \in I_l(j^*) \cup I_h(j^*)$.
    Therefore, we can combine the bounds of \Cref{lem:dual-feasibility1} and \Cref{lem:dual-feasibility2} to obtain
    \begin{align*} 
        &\sum_{S: j^* \in S} \min\{ \redp_{j^*}, \redex(S) \} \cdot \sum_{\substack{k \in I'\\S_k = S}} \frac{|\bar F(k)|}{7 \rho \cdot \redex(S)} \\
    &=
    \sum_{k \in I_l(j^*)} \min\{\redex(S_{k}), \redp_{j^*}\} \cdot \frac{|\bar F(k)|}{7 \rho \cdot \redex(S_k)}
    +
    \sum_{k \in I_h(j^*)} \min\{\redex(S_{k}), \redp_{j^*}\} \cdot \frac{|\bar F(k)|}{7 \rho \cdot \redex(S_k)} \\
    &\leq 8+7 = 15 \ . \qedhere
    \end{align*}
\end{proof}

\subsection{Proof of \Cref{thm:main}}

Finally, we can complete the analysis of Balanced MLF and prove \Cref{thm:main}.

\begin{proof}[Proof of \Cref{thm:main}]
    Let $\lppart^*$ denote the optimal objective value of $(\lppart)$.
    Then we have
    \begin{align}
        |\opt(\ts)| 
        &\geq \lppart^* \label{eq:df:1} \\
        &\geq \sum_{S \subseteq J} \redex(S) \cdot \frac{y_S}{15} \label{eq:df:2} \\
        &\geq \frac{1}{840\rho} |A(\ts)| - \frac{1+14\rho}{210\rho} \ , \label{eq:df:3}
    \end{align}
    where \eqref{eq:df:1} follows from \Cref{lem:relaxation},
    \eqref{eq:df:2} uses weak duality and that $(y_S/15)$ is a feasible solution to $(\dppart)$ by \Cref{lem:dual-feasibility},
    and \eqref{eq:df:3} follows from \Cref{lem:dual-obj}.
    Thus, we conclude that $|A(\ts)| \leq 840\rho \cdot |\opt(\ts)| + 4(1 + 14\rho) \leq 1068\rho \cdot |\opt(\ts)|$, and the algorithm is locally $O(\rho)=O(\nicefrac{\mu_1\mu_2}{\eps})$-competitive, which implies the theorem.
\end{proof}

\section{Randomized Lower Bound}
\label{sec:lower-bounds}

We show that our upper bound is asymptotically tight, even for randomized algorithms, and prove \Cref{thm:lower-bound}.
Our hard instance has $\mu_1 = 1$ and $\mu_2 = \mu$. Thus, throughout this section, we consider the $(\eps,1,\mu)$-online estimate model: every estimate satisfies
$\hat p_j \le p_j \le \mu\cdot \hat p_j$.
The proof extends 
 for an arbitrary $\mu_1$ and $\mu_2$ (where $\mu_1 \mu_2 = \mu$) by multiplying all the estimates by $\mu_1$.

For a schedule $S$, let $\delta_S(t,x)$ denote the number of jobs that are active in $S$ at time $t$ and have remaining processing time at least $x$.
When the schedule is clear from context, we write $\delta(t,x)$.
We write $\delta^*(t)$ for the number of active jobs in an optimal offline schedule at time $t$.
The following standard ``bombardment'' lemma converts a lower bound on the local competitive ratio into a lower bound for total flow time, see e.g., \cite{MotwaniPT94,AzarLT22}. 

\begin{lemma}[Bombardment]
\label{lem:bombardment}
Let $T$ be a fixed time and let $\mathcal D$ be a distribution over instances that releases jobs before time $T$.
Suppose that, for some $\alpha \ge 1$, every deterministic algorithm $A$ satisfies
\[
    \EX_{\mathcal D}[\delta_A(T,1)]
    \ge \alpha \cdot \bigl(\mathbb E_{\mathcal D}[\delta^*(T)] + 1\bigr).
\]
Then there is another distribution over instances on which every deterministic algorithm has expected competitive ratio at least $\Omega(\alpha)$.
\end{lemma}

\begin{proof}
Append to every realization of $\mathcal D$ the same bombardment sequence: for $M$ consecutive integer times $T,T+1,\ldots,T+M-1$, release one unit-size job with $q_j=0$ and exact estimate $\hat p_j=p_j=1$.
These additional jobs are feasible even for distortion $1$.

Fix one such augmented realization.
During the bombardment interval, consider only the $\delta_A(T,1)$ jobs released before time $T$ and the bombardment unit jobs released after time $T$.
If by time $T+s$ the algorithm has completed $z$ of the old jobs, then it has spent at least $z$ units of processing on old jobs.
Among the $s$ bombardment unit jobs released so far, it can therefore have completed at most $s-z$ of them.
Hence at least $\delta_A(T,1)-z$ old jobs and at least $z$ unit jobs are still active, so the algorithm has at least $\delta_A(T,1)$ active jobs throughout the bombardment interval.
Their contribution to the objective during this interval is at least $M\cdot \delta_A(T,1)$.

The offline schedule follows an optimal schedule for the prefix until time $T$, and then processes each new unit job immediately when it is released.
During the bombardment interval it has at most $\delta^*(T)+1$ active jobs.
For the offline schedule, the contribution before time $T$ and after time $T+M$ is independent of $M$; for the algorithm, we only use the contribution during the bombardment interval.
Taking expectations and then letting $M\to\infty$ gives the claim.
\end{proof}

\begin{theorem}
\label{thm:randomized-lower-bound}
For every $\eps\in(0,1]$ and every $\mu\ge 1$, every randomized algorithm for the $(\eps,1,\mu)$-online estimate model has a competitive ratio of at least $\Omega(\nicefrac{\mu}{\eps})$ for minimizing total flow time on a single machine.
\end{theorem}

\paragraph{Intuition for the Lower Bound.} 
To prove the $\Omega(\nicefrac{\mu}{\eps})$ lower bound against any randomized online algorithm, we use Yao's Minimax Principle and design a randomized probability distribution of input jobs. The adversary's strategy relies on two main concepts: 

\begin{enumerate}
    \item \textbf{Information Starvation via ``Coin Flips'':} We release a large batch of jobs at time $0$, whose true processing times are drawn from a Geometric distribution. Because of the online estimate model, the algorithm cannot see a job's true size until it is almost finished. Due to the memoryless nature of the Geometric distribution, all unfinished jobs look completely identical to the algorithm. Working on a job is equivalent to flipping a fair coin: with probability $\nf12$ the job requires more work, and with probability $\nf12$ it reveals its estimate and finishes in exactly one step. 
    
  \item \textbf{Queue Size Discrepancy at a Fixed Time:} We freeze time at a specific moment $T$. Because the online algorithm is ``blindly flipping coins,'' statistical concentration bounds (Hoeffding's inequality) show it is extremely unlikely to find enough ``Heads'' to clear its queue. Meanwhile, the optimal offline schedule (OPT) knows all job sizes at $t=0$ and aggressively clears the shortest jobs first. By time $T$, OPT has finished almost everything, leaving only a tiny handful of massive jobs (Chebyshev's inequality). This creates a massive gap in active queue sizes at time $T$ of size $\Omega(\log n)$. Because the adversary maliciously set the total number of jobs to be exponentially large ($n \approx 2^{\nicefrac{\mu}{\eps}}$), this $\Omega(\log n)$ gap is mathematically equivalent to the desired $\Omega(\nicefrac{\mu}{\eps})$ gap. The Bombardment Lemma (\Cref{lem:bombardment}) then seamlessly converts this local queue gap into a total flow time lower bound.
\end{enumerate}

\begin{proof}[Proof of \Cref{thm:lower-bound}]
Let $P := \nicefrac{\mu}{\eps}$. If $P < 136$, the claim is trivial after adjusting constants, so assume $P \ge 136$. To avoid floor and ceiling notation, assume $\nicefrac{1}{\eps}$ and $\nf{P}{8}$ are integers; standard rounding changes only constants. 
Let $n := 2^{P/8}$ and $L := n^{3/4}$. 

\paragraph{Step 1: The Randomized Instance.}
We define a distribution $\mathcal{D}$ over $n$ jobs, all released at time $0$. We will analyze the expected local competitive ratio at time $T := 3(n-L)$. 

To define $\mathcal{D}$, we first define an auxiliary distribution $\mathcal{D}_0$. For each job $j \in [n]$, independently sample a geometric random variable $Y_j$ with probability $\prob[Y_j = y] = 2^{-y}$, and set the true processing time to $P_j := Y_j + 1$. Thus, $\EX[P_j] = 3$ and $\VAR(P_j) = 2$.

Let $\mathcal{E}$ be the event that the maximum job size is bounded by $P$, i.e., $\mathcal{E} := \{ \max_j P_j \le P \}$. We define our final distribution $\mathcal{D}$ as $\mathcal{D}_0$ conditioned on $\mathcal{E}$. A union bound shows this complementary event is extremely unlikely:
\begin{equation}
\label{eq:conditioning-event}
    \prob_{\mathcal{D}_0}[\overline{\mathcal{E}}] \le n \cdot \prob[P_j > P] \le n \cdot 2^{-(P-1)} \le 2^{1-7P/8}.
\end{equation}

For every realized job $j$ in the sampled instance, we set the online estimate threshold $q_j$ and the estimate $\hat{p}_j$ as follows:
\[
    q_j := \min\{P_j-1, \nicefrac{1}{\eps}-1\} \qquad \text{and} \qquad \hat{p}_j := \min\{P_j, \nicefrac{1}{\eps}\}.
\]

\emph{Feasibility Check:} We verify this instance is strictly feasible for the $(\eps, 1, \mu)$ model:
\begin{itemize}
    \item If $P_j \le \nicefrac{1}{\eps}$, then $q_j = P_j - 1 \le P_j - \eps(\nicefrac{1}{\eps}) \le (1-\eps)P_j$. 
    \item If $P_j > \nicefrac{1}{\eps}$, then $P_j \ge \nicefrac{1}{\eps} + 1$, and $q_j = \nicefrac{1}{\eps} - 1 = (1-\eps) \cdot \nicefrac{1}{\eps} < (1-\eps)P_j$. 
    \item Under event $\mathcal{E}$, jobs with $P_j > \nicefrac{1}{\eps}$ satisfy $\hat{p}_j = \nicefrac{1}{\eps} \le P_j \le P \le \mu(\nicefrac{1}{\eps}) = \mu \hat{p}_j$. Jobs with $P_j \le \nicefrac{1}{\eps}$ are predicted exactly. Thus, all estimates are feasible.
\end{itemize}

\paragraph{Step 2: Lower Bounding the Algorithm's Queue.}
Because the job sizes $P_j$ are based on a geometric distribution $Y_j$, we can represent $Y_j$ as the index of the first ``Heads'' in a sequence of independent fair coin flips. For the algorithm, processing the first $Y_j$ units of job $j$ is equivalent to exposing coin flips until a Heads appears. Once a Heads appears, the job requires exactly one final deterministic unit of work to complete. (We may  reveal this Heads to the algorithm immediately when it occurs, as extra information only helps the algorithm).

Let $m := n - \nf{L}{2}$. For algorithm $A$ to complete at least $m$ jobs by time $T$, it must encounter $m$ Heads. Since each completed job requires $1$ deterministic final unit of work, these $m$ final units consume $m$ steps of time. This leaves at most $s := T - m = 2n - \nf{5L}{2}$ steps available for exposing coin flips. 

Algorithm $A$ therefore needs to find at least $m = n - \nf{L}{2}$ Heads in $s$ flips. For a binomial variable $H \sim \mathrm{Bin}(s, \nf12)$, the expected number of Heads is $\EX[H] = \nf{s}{2} = n - \nf{5L}{4}$. Requiring $m$ Heads represents an  upward deviation of $\nf{3L}{4}$. By Hoeffding's inequality (using $s \le 2n$):
\begin{equation}
\label{eq:chernoff-algorithm}
    \prob_{\mathcal{D}_0}[A \text{ completes } \ge m \text{ jobs by time } T] \le \exp\left(-\frac{2(3L/4)^2}{s}\right) \le \exp\left(-\frac{9}{16} n^{1/2} \right).
\end{equation}

Consequently, for all sufficiently large $P$, the expected number of active jobs remaining for $A$ under $\mathcal{D}_0$ is:
\[
    \EX_{\mathcal{D}_0}[\delta_A(T,1)] \ge \frac{L}{2}\left(1-\exp\left(-\frac{9}{16} n^{1/2} \right)\right) \ge \frac{L}{4}.
\]
Conditioning on $\mathcal{E}$ loses at most $n \cdot \prob_{\mathcal{D}_0}[\overline{\mathcal{E}}]$ in expectation. Since this loss is $o(L)$ by \Cref{eq:conditioning-event}, we obtain:
\begin{equation}
\label{eq:algorithm-local-lb}
    \EX_{\mathcal{D}}[\delta_A(T,1)] \ge \frac{L}{8}.
\end{equation}

\paragraph{Step 3: Upper Bounding OPT's Queue.}
Unlike the algorithm, OPT knows all true job sizes at $t=0$ and schedules jobs using Shortest Processing Time (SPT). 

To analyze OPT's queue, we rely on Chebyshev's inequality, which, as we recall, states that for any random variable $X$ with expected value $\EX[X]$ and variance $\VAR(X)$, the probability that $X$ deviates from its mean by at least $k$ is bounded by:
\[
    \prob\big[|X - \EX[X]| \ge k\big] \le \frac{\VAR(X)}{k^2}.
\]

\textbf{Step 3.1: Bounding the Total Work.} 
Let $W := \sum_{j=1}^n P_j$ be the total work in the system. Because the job sizes $P_j$ are independent with $\EX[P_j] = 3$ and $\VAR(P_j) = 2$, linearity of expectation and variance gives $\EX[W] = 3n$ and $\VAR(W) = 2n$. 

We want to bound the probability that the total work is excessively large. Applying Chebyshev's inequality with $k = L$:
\begin{equation}
\label{eq:total-work-bound}
    \prob_{\mathcal{D}_0}[W > 3n+L] \le \prob_{\mathcal{D}_0}\big[|W - 3n| \ge L\big] \le \frac{\VAR(W)}{L^2} = \frac{2n}{L^2} = 2n^{-1/2}.
\end{equation}

\textbf{Step 3.2: Bounding the Number of Massive Jobs.} 
Let $b := \frac{1}{4}\log_2 n$, and let $B$ be the number of ``massive'' jobs with size $P_j > b$. We can write $B$ as the sum of $n$ independent indicator variables, $B = \sum_{j=1}^n I_j$, where $I_j = 1$ if $P_j > b$ and $0$ otherwise. 

First, we calculate the probability that a single job is massive:
\[
    \prob[P_j > b] = \prob[Y_j \ge b] = 2^{-(b-1)} = 2 \cdot 2^{-\frac{1}{4}\log_2 n} = 2n^{-1/4}.
\]
Using this, we find the expected number of massive jobs (recall $L = n^{3/4}$):
\[
    \EX[B] = n \cdot \prob[P_j > b] = n \left(2n^{-1/4}\right) = 2n^{3/4} = 2L.
\]
Because the jobs are independent, the variance of $B$ is the sum of the variances of the indicator variables. For any Bernoulli variable, the variance $p(1-p)$ is strictly bounded by $p$. Thus:
\[
    \VAR(B) = n \cdot p(1-p) \le np = \EX[B] = 2L.
\]
We want to bound the probability of the ``bad'' event where we get too few massive jobs, specifically $\prob[B < L]$. Notice that if $B < L$, its distance from the mean ($2L$) is strictly greater than $L$. Applying Chebyshev's inequality with $k = L$ yields our second bound:
\begin{equation}
\label{eq:many-large-jobs}
    \prob_{\mathcal{D}_0}[B < L] \le \prob_{\mathcal{D}_0}\big[|B - 2L| \ge L\big] \le \frac{\VAR(B)}{L^2} \le \frac{2L}{L^2} = \frac{2}{L}.
\end{equation}

\textbf{Step 3.3: Bounding OPT's Active Jobs at Time $T$.} 
Assume the highly probable event that $W \le 3n+L$ and $B \ge L$ holds. At time $T$, the remaining volume of work for OPT is at most $W - T \le (3n+L) - 3(n-L) = 4L$. 

Because OPT prioritizes the shortest jobs, it is impossible for OPT to have any job of size $\le b$ still active at time $T$. If it did, it would mean all $B \ge L$ massive jobs (size $> b$) must also still be in the queue, which would require a remaining work volume strictly greater than $L \cdot b$. Since $P \ge 136$, we have $n = 2^{P/8} \ge 2^{17}$, which means $b = \frac{1}{4}\log_2 n \ge 4.25 > 4$. Hence, $Lb > 4L$, which gives a contradiction.

Therefore, every active job in OPT's queue at time $T$ (except possibly the single job currently executing) must be a massive job of size $> b$. Since the total remaining work is at most $4L$, the \emph{number} of such active jobs is strictly bounded:
\[
    \delta^*(T) \le 1 + \frac{4L}{b}.
\]
Factoring in the trivial bound ($\delta^*(T) \le n$) on the complementary bad events using \Cref{eq:total-work-bound} and \Cref{eq:many-large-jobs}, we obtain for sufficiently large $P$:
\begin{equation}
\label{eq:opt-local-d0}
    \EX_{\mathcal{D}_0}[\delta^*(T)] \le 1 + \frac{4L}{b} + n\left(\frac{2}{\sqrt{n}} + \frac{2}{L}\right) \le \frac{6L}{b}.
\end{equation}
Finally, \Cref{eq:conditioning-event} implies $\prob_{\mathcal{D}_0}[\mathcal{E}] \ge \nf12$ for sufficiently large $P$, and hence:
\begin{equation}
\label{eq:opt-local-d}
    \EX_{\mathcal{D}}[\delta^*(T)] \le 2\EX_{\mathcal{D}_0}[\delta^*(T)] \le \frac{12L}{b}.
\end{equation}

\paragraph{Step 4: Conclusion.}
Combining \Cref{eq:algorithm-local-lb} and \Cref{eq:opt-local-d} to compare the local queue sizes at time $T$, we get:
\[
    \frac{\EX_{\mathcal{D}}[\delta_A(T,1)]}{\EX_{\mathcal{D}}[\delta^*(T)]+1} 
    \ge \Omega(b) 
    = \Omega(\log n) 
    = \Omega(P) 
    = \Omega\left(\frac{\mu}{\eps}\right).
\]

By the Bombardment Lemma (\Cref{lem:bombardment}), appending a unit-job bombardment sequence to this distribution converts this expected local queue size ratio into an expected total flow time ratio of $\Omega(\nicefrac{\mu}{\eps})$ for any deterministic algorithm. Yao's principle immediately extends this lower bound to every randomized algorithm.
\end{proof}

\section*{Acknowledgments}
    This research was partially supported by Dr. Max
Rössler, the Walter Haefner Foundation, and the ETH Zürich Foundation. Parts of these results were established 
    when SY was visiting TU Berlin in March 2026. This visit
    was supported by the Deutsche Forschungsgemeinschaft (DFG, German Research Foundation) under Germany’s Excellence Strategy -- The Berlin Mathematics Research Center MATH$^+$ (EXC-2046/2). 
    Part of the work was done while JS was affiliated with Centrum Wiskunde \& Informatica (CWI), and supported by the Netherlands Organisation for Scientific Research (NWO) through project OCENW.GROOT.2019.015 ``Optimization for and with Machine Learning (OPTIMAL)''.
    Research of HK was partially supported by the Israel Science Foundation (ISF) grant no.\ 1156/23 and the Blavatnik research Foundation.

\printbibliography

\newpage
\appendix

\section{Reducing $\eps$-Signaling to Online Estimates}
\label{app:reduction}

We briefly argue that the $\eps$-signaling model is a special case of the online estimate model. 

\begin{lemma}
    For every $\eps \in (0,\nf12]$,
    every $f(\eps,\mu_1,\mu_2)$-competitive algorithm for minimizing total flow in the $(\eps,\mu_1,\mu_2)$-online estimate model is $f(\eps, 1, \nf{1}{\eps})$-competitive for minimizing total flow in the $\eps$-signaling model.
\end{lemma}

A consequence of this lemma is that \Cref{thm:main}
implies \Cref{thm:signal}.

\begin{proof}
The main idea is that the ``estimate'' in the $\eps$-signaling model is simply the elapsed time of the job when the signal arrives. Hence, an algorithm for online estimates can be used for the $\eps$-signaling model by using $\hat p_j = e_j(s_j)$ for each job $j$, where $s_j$ is the time when job $j$ gives the signal.

We verify that this is indeed a feasible instance for $(\eps,\mu_1,\mu_2)$-online estimates, where $\mu_1 = 1$ and $\mu_2 = \nf{1}{\eps}$.
First, the estimate does not arrive too late, that is, $e_j(s_j) \leq (1-\eps) p_j$, which is just an assumption of the $\eps$-signaling model.
Second, the distortion is bounded: we clearly have $\hat p_j = e_j(s_j) \leq  p_j = \mu_1 p_j$, 
and further it holds that  $\hat p_j = e_j(s_j) \geq \eps p_j = \frac{1}{\mu_2} p_j$, where we use that the signal does not arrive too early.
\end{proof}

\section{Bounding the Number of Preemptions of Balanced MLF}
\label{app:preemptions}

\begin{lemma}\label{lem:balanced-mlf-preemptions}
    The total number of preemptions of Balanced MLF is in
    $O(\sum_j (1+\log_2 p_j))$.
\end{lemma}
\begin{proof}
We count a preemption whenever a job $j$ is processed during $[t,t+1]$,
is unfinished at time $t+1$, and is not processed during $[t+1,t+2]$.
We distinguish two cases.

First, suppose that the current class of $j$ increases at time $t+1$.
Then Balanced MLF makes $j$ fresh again, and we charge the preemption to
this class increase.  For each job $j$, the current class can increase
at most once when the estimate $\hat p_j$ arrives and at most
$O(1+\log_2 p_j)$ times due to processing.  Thus, the total number of
preemptions of this type is $O(\sum_j (1+\log_2 p_j))$.

It remains to consider a preemption of a job $j$ whose current class does
not increase at time $t+1$. Then $j$ can stop being processed only
because, at time $t+1$, some fresh job $j'$ of strictly smaller current
class becomes ongoing. We claim that this can happen only at a time when
a job is released.

Suppose for the sake of contradiction that no job is released at time
$t+1$. Since estimates either arrive at the release time of a job or
after the job has been processed, the only estimate of an already
released job that can arrive at time $t+1$ is the estimate of the job
processed during $[t,t+1]$, namely $j$. If the arrival of this estimate
increased the current class of $j$, then the preemption was already
counted in the first case. Hence, in the present case, no estimate
arrival can create a new fresh job of smaller current class.

Moreover, no job other than $j$ is processed during $[t,t+1]$. Thus no
other already released job changes its current class due to processing
during this time step, and no other already released job can receive an
estimate at time $t+1$. Since $j$ is unfinished and its current class
does not increase, $j$ remains ongoing at time $t+1$. Therefore every
fresh job of smaller current class than $j$ that is active at time
$t+1$ was already fresh and of smaller current class than $j$ at time
$t$.

As Balanced MLF processed $j$ during $[t,t+1]$, such a fresh job was not
made ongoing at time $t$. Hence the balance condition must have failed
at time $t$, that is, $|\afull(t)| < \frac14 |A(t)|$.

In the absence of a job release at time $t+1$, the number of fresh jobs
cannot increase between the scheduling decisions at times $t$ and
$t+1$. Also, since the processed job $j$ is unfinished, the number of
active jobs cannot decrease. Hence the balance condition cannot become
true at time $t+1$. This contradicts the assumption that a fresh job
$j'$ of smaller current class becomes ongoing at time $t+1$.

Thus, every preemption of the second type occurs at a time at which a job
is released. We charge such a preemption to a job released at the same
time. Since there is at most one processed job before any time step,
there is at most one preemption at any time. Breaking ties between
simultaneous releases arbitrarily, these preemptions can therefore be
charged to distinct job releases. Hence the number of preemptions of
the second type is at most $n$.

Combining the two bounds gives
\[
    O(n) + O\Bigl(\sum_j (1+\log_2 p_j)\Bigr)
    =
    O\Bigl(\sum_j (1+\log_2 p_j)\Bigr),
\]
as claimed.
\end{proof}

Note that the lemma only bounds the total number of preemptions: a single
job may be preempted up to $n-1$ times.

\section{Omitted Proofs}\label{sec:proofs}

\stackmonotone*
\begin{proof}
  By the rule of the algorithm, we only make a job ongoing if it has a
  smaller current class than the current class of any ongoing job. Whenever the current class of an ongoing job changes, it
  becomes fresh. Hence, at any time, the current classes of ongoing
  jobs are pairwise distinct.
\end{proof}

\nongreedy*
\begin{proof} 
  Let $k := k_j(t).$ By \Cref{lem:stack-monotone} and the processing
  rule of the algorithm, every active job of class at most $k$ other
  than $j$ must be fresh.  Since the algorithm processes $j$ during
  $[t,t+1]$,
  by \Cref{obs:1}
  it did not make $j_1$ ongoing at time $t$ because
  \begin{equation}	
    |\afull(t)| < \frac14 |A(t)|.          \label{eq:balance-condition-wrong}
  \end{equation}
  In particular, $j$ cannot be the single ongoing job at time $t$: if
  it were, then $|A(t)| = |\afull(t)|+1$, and the existence of $j_1$
  would imply $|\afull(t)|\geq |A(t)|/4$, contradicting
  \eqref{eq:balance-condition-wrong}.  Let $j' \in \aon(t)$ be the
  ongoing job of the next largest current class at time $t$ after $k$,
  and let $t'$ be the last time before $t$ at which $j'$ became
  ongoing.

  Suppose for the sake of contradiction that $i \in A(t)$ is a
  job other than $j$ and $j_1$, with $k_i(t) \le k_j(t)$. When $j'$
  became ongoing at time $t'$, it was a smallest-class active job.
  The class of $j'$ has not changed after time $t'$, since it had a
  larger current class than $j$ when $j$ became ongoing, and after
  that has not been processed anymore.  Therefore, if any of $j,j_1,i$
  had already been active at time $t'$, monotonicity of current
  classes would imply that it had class smaller than the class of $j'$
  at time $t'$, contradicting the choice of $j'$. Hence $j,j_1,i$ were
  all released after $t'$.

Let $J'$ be the jobs active at time $t$ that were released after $t'$.
Every job in $J'\setminus\{j\}$ is fresh at time $t$: if such a job became ongoing after $t'$, it would have a smaller current class than $j'$; since $j'$ 
is an ongoing job with the next largest current class after $j$ at time $t$, and $j$ is processed at time $t$, $j'$ and $j$ are the two ongoing jobs of smallest current classes.

Immediately before $j'$ became ongoing at time $t'$, the balance condition
$|\afull|\ge |A|/4$ held. From that moment to time $t$, the jobs in
$J'\setminus\{j\}$ contribute $|J'| - 1$ new fresh jobs and $|J'|$ new
active jobs, while making $j'$ ongoing contributes $-1$ to the number of
fresh jobs and $0$ to the number of active jobs. Thus the left-hand side
of the balance inequality changes by $|J'| - 2$, and the right-hand side
changes by $|J'|/4$. Since $j,j_1,i \in J'$, we have $|J'|\ge 3$, and
therefore
\[
    |J'| - 2 \ge \frac14 |J'|.
\]
So the balance condition would still hold at time $t$, contradicting
\eqref{eq:balance-condition-wrong}.
\end{proof}

\manyfresh*
\begin{proof}
  The proof is by induction. At time $t=0$, 
  we have $|A(0)| = 0$ and
  $|\afull(0)| = 0$, so the claim holds trivially. Now, at some time
  $t$, a job release increases both $|\afull(t)|$ and $|A(t)|$ by one,
  so it preserves the inequality. 
  Making a fresh job ongoing only
  happens if $|\afull(t)|\ge \frac14 |A(t)|$. After that $|\afull(t)|$
  decreases by one and $|A(t)|$ is unchanged, so
  $|\afull(t)|\ge \frac14|A(t)| - 1$. A job completion decreases $|A(t)|$ and
  cannot hurt the inequality. Making an ongoing job fresh again after
  it got promoted increases $|\afull(t)|$. Thus the invariant is
  preserved.
\end{proof}

\manyinbarF*
\begin{proof}
    In each iteration in the construction of the sets $\bar{F}$, there is a largest effective class $k_i$ that satisfies Conditions (i), (ii) and (iii) of the second step in the construction and leads to the removal of some job $j$ in the subsequent third step. 
    We say that $k_i$ \emph{triggers} the removal of a job.
    Note that each such $k_i$ can trigger the removal of at most one job $j$, because there is only one interval $[t_{\ge k_i},t_{\ge k_i}+1)$.

    In order to prove the statement of the lemma, we first claim that if $k_i$ triggers the removal of a job $j$, then the effective class of $j$ is strictly smaller than $k_i$, i.e., $\redk_{j} < k_i$. To see this, recall that by~\Cref{def:effective:class}, the effective class of $j$ is $\redk_{j} = \min\{\fink_j,\max_{\fint_j \le t < \ts} k(t)\}$. By Condition (ii) of the second step in the construction, $j$ reaches its final class at time $\fint_j = t_{\ge k_i} + 1$. Using the definition of $t_{\ge k_i}$, we can conclude
    $$
    \redk_{j} = \min \Bigl\{\fink_j,\max_{\fint_j \le t < \ts} k(t) \Bigr\} \le \max_{\fint_j \le t < \ts} k(t) = \max_{t_{\ge k_i} + 1 \le t < \ts} k(t) < k_i.
    $$

    To prove the statement of the lemma, let $k_1 > \ldots > k_{\ell}$ denote all effective classes that trigger the removal of a job, where the strict inequalities are justified by the fact that each effective class can trigger the removal of at most one job.    
    That is, a total of $\ell$ jobs are removed and $|\bigcup_k \bar{F}(k)| = |\bigcup_k F(k)| - \ell$. We show that $|\bigcup_k \bar{F}(k)| \ge \ell$, which then implies $|\bigcup_k F(k)| \ge 2 \cdot \ell$ and thus the lemma.

    To this end, we show by induction that for all $1 \le d \le \ell$  
    \[
    \Bigl|\bigcup_{k \ge k_d} \bar{F}(k) \Bigr| \ge d.
    \]
    For $d=1$, we have that $k_d = k_1$ is the largest class that triggers the removal of a job in the first iteration of the construction. Since $k_d$ satisfies Condition (i) of the second step, the set $\bar{F}(k_d)$ was not empty in that iteration. We already argued that each effective class $k_i$ only triggers the removal of a job $j$ with a strictly smaller effective class. Hence, the element of $\bar{F}(k_d)$ is not removed in any iteration of the construction. We can conclude the base case with $|\bigcup_{k \ge k_1} \bar{F}(k)| \ge 1$.

    For the induction step, we assume that $|\bigcup_{k \ge k_d} \bar{F}(k)| \ge d$ and show that $|\bigcup_{k \ge k_{d+1}} \bar{F}(k)| \ge d+1$. Since the effective class $k_{d+1}$ satisfies Condition (i) of the construction, we have that $\bar{F}(k_{d+1}) \not= \emptyset$ at the beginning of the iteration where $k_{d+1}$ triggers the removal of a job. Following the same reasoning as in the base case, we know that members of $\bar{F}(k_{d+1})$ are not removed during the iteration or any subsequent iteration. Thus, $\bar{F}(k_{d+1}) \not= \emptyset$ also holds at the end of the construction. Since $k_{d+1} < k_d$ implies $\bar{F}(k_{d+1}) \cap (\bigcup_{k \ge k_d} \bar{F}(k)) = \emptyset$, we can conclude that
    \[ 
    \Bigl|\bigcup_{k \ge k_{d+1}} \bar{F}(k) \Bigr| \ge d+1
    \]
    by using $\bar{F}(k_{d+1}) \not= \emptyset$ and the induction hypothesis $|\bigcup_{k \ge k_d} \bar{F}(k)| \ge d$.
\end{proof}

\end{document}